\pdfoutput=1
\newcommand{\extended}[1]{}    %
\newcommand{\short}[1]{#1}     %

\documentclass[sigconf]{aamas}

\usepackage{balance} %

\usepackage[T1]{fontenc}
\usepackage[utf8]{inputenc}
\usepackage{graphicx}

\usepackage{kim}
\usepackage{wojtek15logics}
\usepackage{wojtek15other}

\usepackage{tabularx}
\usepackage{prooftree}
\usepackage{layouts}

\usepackage{bbm}

\usepackage[only,llbracket,rrbracket]{stmaryrd}
\newcommand{\sset}[1]{\llbracket{#1}\rrbracket}

\graphicspath{ {./images/} }

\makeatletter
\g@addto@macro{\UrlBreaks}{\UrlOrds}
\makeatother

  \makeatletter
  \DeclareFontEncoding{LS2}{}{\@noaccents}
  \makeatother
  \DeclareFontSubstitution{LS2}{stix}{m}{n}

  \DeclareSymbolFont{largesymbolsstix}{LS2}{stixex}{m}{n}

  \DeclareMathDelimiter{\lBrace}{\mathopen}{largesymbolsstix}{"E8}{largesymbolsstix}{"0E}
  \DeclareMathDelimiter{\rBrace}{\mathclose}{largesymbolsstix}{"E9}{largesymbolsstix}{"0F}
  \newcommand{\multiset}[1]{\lBrace{#1}\rBrace}

\setcopyright{ifaamas}
\acmConference[AAMAS '25]{Proc.\@ of the 24th International Conference
on Autonomous Agents and Multiagent Systems (AAMAS 2025)}{May 19 -- 23, 2025}
{Detroit, Michigan, USA}{Y.~Vorobeychik, S.~Das, A.~Nowe (eds.)}
\copyrightyear{2025}
\acmYear{2025}
\acmDOI{}
\acmPrice{}
\acmISBN{}

\acmSubmissionID{929}

\title[Practical Abstractions for Continuous-Time Multi-Agent Systems]{Practical Abstractions for Model Checking \\Continuous-Time Multi-Agent Systems}

\author{Yan Kim}
\affiliation{
  \institution{\instSNT}
  \city{Esch-sur-Alzette}
  \country{Luxembourg}%
}
\email{yan.kim@uni.lu}

\author{Wojciech Jamroga}
\affiliation{
  \institution{\instIPI}
  \city{Warsaw}
  \country{Poland}\\
  \institution{\instSNT}
  \city{Esch-sur-Alzette}
  \country{Luxembourg}
}
\email{w.jamroga@ipipan.waw.pl}

\author{Wojciech Penczek}
\affiliation{
  \institution{\instIPI}
  \city{Warsaw}
  \country{Poland}%
}
\email{w.penczek@ipipan.waw.pl}

\author{Laure Petrucci}
\affiliation{
  \institution{\instLIPN}
  \city{Villetaneuse}
  \country{France}%
}
\email{laure.petrucci@lipn.univ-paris13.fr}

\begin{abstract}
  Model checking of temporal logics in a well established technique to verify and validate properties of multi-agent systems (MAS). However, practical model checking requires input models of manageable size. In this paper, we extend the model reduction method by variable-based abstraction, proposed recently by Jamroga and Kim, to the verification of real-time systems and properties. To this end, we define a real-time extension of MAS graphs, extend the abstraction procedure, and prove its correctness for the universal fragment of Timed Computation Tree Logic (TCTL). Besides estimating the theoretical complexity gains, we present an experimental evaluation for a simplified model of the Estonian voting system and verification using the Uppaal model checker.
\end{abstract}


\newcommand{\BibTeX}{\rm B\kern-.05em{\sc i\kern-.025em b}\kern-.08em\TeX}
\newcommand{\memout}{\textsc{memout}}

\newcommand{\Rplus}{\mathbb{R}_{+}}

\newcommand{\Clocks}{\mathcal{X}}
\newcommand{\Vars}{\mathcal{V}}
\newcommand{\Inv}{\mathcal{I}}
\newcommand{\Constr}{\mathcal{C}}

\newcommand{\Loc}{\mathit{Loc}}
\newcommand{\Eval}{\mathit{Eval}}
\newcommand{\Chan}{\mathit{Chan}}
\newcommand{\Sync}{\mathit{Sync}}
\newcommand{\Act}{\mathit{Act}}
\newcommand{\Cond}{\mathit{Cond}}
\newcommand{\Effect}{\mathit{Effect}}
\newcommand{\Edges}{E}
\newcommand{\ag}[1]{^{#1}}

\newcommand{\Paths}{\mathit{Paths}}
\newcommand{\Reach}{\mathit{Reach}}

\newcommand{\States}{\mathit{St}}
\newcommand{\Dom}{domain}
\newcommand{\clc}{\mathfrak{cc}}

\newcommand{\abstr}{\mathcal{A}}
\newcommand{\may}{\mathit{may}}

\newcommand{\mayabstr}{\abstr^{\may}}

\newcommand{\uext}{_\circleddash} %

\newcommand{\Scope}{\textit{Sc}}
\newcommand{\Argsd}{\mathit{Args}_{R}}
\newcommand{\Argsr}{\mathit{Args}_{N}}

\newcommand{\updcomp}{\circ}
\newcommand{\sini}{\mathit{ini}}

\newcommand{\nMG}{\mathcal{N}}

\newcommand{\myinterval}[1]{{\mbox{\ensuremath{[#1]}}}}
\newcommand{\pii}[1]{\pi\myinterval{#1}}

\newcommand*{\AddNoteSmall}[4]{%
    \begin{tikzpicture}[overlay, remember picture]
        \draw [decoration={brace,amplitude=0.5em},decorate,thick,mdgreen]
            ($(#3)!([yshift=1ex]#1)!($(#3)-(0,1)$)$) --
            ($(#3)!(#2.south)!($(#3)-(0,1)$)$)
                node [inner sep=10pt, pos=0.5, anchor=west, text width=1.2cm] {{\small #4}};
    \end{tikzpicture}
}%

\SetKwFunction{approxLocalDomain}{ApproxLocalDomain}
\SetKwFunction{overApproxLocalDomain}{OverApproxLocalDomain}
\SetKwFunction{visitLoc}{VisitLoc}
\SetKwFunction{procEdgeLabels}{ProcEdge}
\SetKwFunction{procPreCond}{ProcPreCond}
\SetKwFunction{enqueueSucc}{EnqueueDirectSucc}
\SetKwFunction{procSelfLoops}{ProcSelfLoops}
\SetKwFunction{procIncEdges}{ProcIncEdges}
\SetKwFunction{computeAbstraction}{ComputeAbstraction}
\SetKwFunction{computeMergeAbstraction}{ComputeAbstraction}
\SetKwFunction{generalVarAbstraction}{ComputeAbstraction}
\SetKwFunction{extractMax}{ExtractMax}
\SetAlgoLined\DontPrintSemicolon

\newlength{\semanticsnewlinebreak}
\newcommand{\dmax}{d_{\mathit{max}}}

\begin{document}

\pagestyle{fancy}
\fancyhead{}

\maketitle

\section{Introduction}

Temporal logics have been extensively used to formalize properties of agent systems, including reachability, liveness, safety, and fairness~\cite{Emerson90temporal}.
Moreover, temporal model checking is a popular approach to formal verification of MAS~\cite{baier2008principles,Clarke18principles}.
However, the verification is known to be hard, both theoretically and in practice.
State-space explosion is a major obstacle here, as models of real-world systems are huge and infeasible even to generate, let alone verify. In consequence, model checking of MAS w.r.t.~their \emph{modular representations} ranges from \PSPACE-complete to undecidable~\cite{Schnoebelen03complexity,Bulling10verification}.

Much work has been done to contain the state-space explosion by smart representation and/or reduction of input models. Symbolic model checking based on SAT- or BDD-based representations of the state/transition space~\cite{McMillan93symbolic-mcheck,McMillan02unbounded,Penczek03ctlk,Kacprzak04verifying,Lomuscio07tempoepist,Huang14symbolic-epist,lomuscio17mcmas} fall into the former group.
Model reduction methods include partial-order reduction~\cite{Peled93representatives,Gerth99por,Jamroga20POR-JAIR}, equivalence-based reductions~\cite{Bakker84equivalences,Alur98refinement,Belardinelli21bisimulations}, and state abstraction~\cite{Cousot77abstraction}, see below for a detailed discussion.

In this paper, we extend the idea of \emph{variable-based abstraction}~\cite{Jamroga23variableabstraction,Jamroga23easyabstract} to the verification of \emph{real-time multi-agent systems}~\cite{alur1993model,Brihaye07timedCGS,Arias23stratTCTL,Penczek06advances}.
Similarly to~\cite{Jamroga23variableabstraction,Jamroga23easyabstract}, our abstraction operates entirely on the high-level, modular representation of an asynchronous MAS.
That is, it takes a \emph{concrete modular representation} of a MAS as input, and generates an \emph{abstract modular representation} as output. %
Moreover, it produces the abstract representation without generation of the explicit state model, thus avoiding the usual computational bottleneck.

\para{Related Work.}
State abstraction was introduced in~\cite{Cousot77abstraction}, and studied intensively in the context of temporal verification~\cite{Clarke94abstraction,Godefroid02abstraction,Dams18abstraction+refinement,Clarke00cegar,Shoham04abstraction,Clarke00cegar,Clarke03cegar,Cimatti02nusmv}.
However, those works propose lossless abstraction that typically obtain up to an order of magnitude reduction of the state space and output models that are still too large for practical verification. 
Here, we focus on lossy may abstractions, based on user-defined equivalence relations~\cite{Dams97abstraction,Godefroid01abstractionbased,Godefroid02abstraction,Godefroid14abstraction-software,Gurfinkel06yasm,Godefroid10yogi,Enea08abstractions,Cohen09abstraction-MAS,Lomuscio10dataAbstraction}.

May/must abstractions for strategic properties have been investigated in~\cite{Alfaro04three,Ball06abstraction,Kouvaros17predicateAbstraction,Belardinelli17abstraction,Belardinelli19abstractionStrat}.
In all those cases, the abstraction method is defined directly on the concrete model, i.e., it requires to first generate the concrete global states and transitions, which is exactly the bottleneck that we want to avoid.
In contrast, our method operates on modular (and compact) model specifications, both for the concrete and the abstract model.
Data abstraction methods for infinite-state MAS~\cite{Belardinelli11data-abstraction,Belardinelli17dataAware} come close in that respect, but they still generate explicit state models.
Even closer, \cite{Jamroga23variableabstraction,Jamroga23easyabstract}~proposed recently a user-friendly abstraction scheme via removal of variables in the modular agent templates.
However, all the above approaches deal only with the verification of \emph{untimed} models and properties.

Timed models of MAS and their verification have been studied for over 30 years now, see e.g.~\cite{alur1993model,Brihaye07timedCGS,Arias23stratTCTL} and especially~\cite{Penczek06advances} for an overview.
Time-abstracting bisimulation for timed automata was studied in~\cite{tripakis2001analysis}.
In this paper, we extend the ideas and results from~\cite{Jamroga23variableabstraction} to models of real-time asynchronous MAS of~\cite{Arias23stratTCTL}.

The study of MAS using Uppaal was conducted in~\cite{gu2022verifiable} and~\cite{gu2022correctness}. 
However, its authors tackled a different problem --- strategy synthesis, in a different setting --- stochastic timed games. 
In some ways, their methods can be seen as complementary to ours, as they first synthesize the (witness) strategy and then check its validity; whereas we mainly focus on the verification of safety properties.

\section{Reasoning about Real-Time MAS}
\label{sec:prelim}

We start with extending the (untimed) representations of asynchronous MAS in~\cite{Jamroga23variableabstraction} to their timed counterparts.

We first adopt the usual definitions of clocks in timed systems, e.g. \cite{alur1993model}.
Let $\Clocks=\set{x_1,\ldots,x_{n_{\Clocks}}}$ be a finite set of clock variables.
A clock valuation is a mapping $\upsilon:\Clocks\mapsto \Rplus$.\footnote{By $\Rplus$ we denote the set of non-negative real numbers.}
Given a valuation $\upsilon$, a delay $\delta\in\Rplus$ and $X \subseteq \Clocks$,
$\upsilon+\delta$ denotes the valuation $\upsilon'$, such that $\upsilon'(x)=\upsilon(x)+\delta$ for all $x\in\Clocks$, and
$\upsilon[X=0]$ denotes the valuation $\upsilon''$, such that $\upsilon''(x)=0$ for all $x\in X$ and  $\upsilon''(x)=\upsilon(x)$ for all $x\in \Clocks\setminus X$.

The set $\Constr_\Clocks$ of clock constraints over $\Clocks$ is inductively defined by the following grammar:
\quad
\[
  \clc ::= \top\ \mid\ x_i \sim c\ \mid\ x_i-x_j \sim c\ \mid\ \clc \wedge \clc,
\]
\noindent%
where $\top$ denotes the truth value, $\forall_{i,j\in\set{1,\ldots, n_\Clocks}} x_i,x_j\in\Clocks$, $c\in\mathbb{N}$, and ${\sim}\in\set{<,\leq,=,\geq,>}$.

Let $\Vars=\set{v_1,\ldots,v_{n_{\Vars}}}$ be a finite set of discrete (typed) numeric variables over finite domains. %
By $\Eval_\Vars$ we denote the set of evaluations over the set of (discrete) variables $\Vars$; an evaluation $\eta\in\Eval_\Vars$ maps every variable $v\in\Vars$ to a literal from their domain $\eta(v)\in \Dom(v)$. 
Expressions are constructed from variables $\Vars$ and literals $\bigcup_{i=1}^{n_{\Vars}}\Dom(v_i)$ using the arithmetic operators. 
Atomic formulas/conditions are built from expressions and relation symbols. 
They can be further combined with logical connectives to form predicates. 
The set of all possible predicates over the variables $\Vars$ is denoted by $\Cond_\Vars$. 

The sets of all valuations satisfying $\clc\in\Constr_\Clocks$ and all evaluations satisfying $g\in\Cond_\Vars$ shall be denoted by $\sset{\clc}$ and $\sset{g}$ respectively. 
For $v\in\Vars$, $k\in\Dom(v)$, by $g[v=k]$ we denote the substitution of all occurrences of the variable $v$ in $g$ with the literal $k$; analogously, for $V\subseteq \Vars$, $K=\set{k_{v}\in\Dom(v) \mid v\in V}$, by $g[V=K]$ we denote $g[v=k_v \mid v\in \Vars]$.

For simplicity, we assume that both $\Clocks$ and $\Vars$ have (some) ordering fixed.
In the sequel, the terms \emph{variables} and \emph{clocks} will mean \emph{discrete variables} and \emph{continuous variables} respectively. %

{\setlength{\abovecaptionskip}{0.1ex}%
\setlength{\belowcaptionskip}{0.2ex}%
\begin{figure}[!tbh]
	\centering
	\includegraphics[width=0.69\columnwidth]{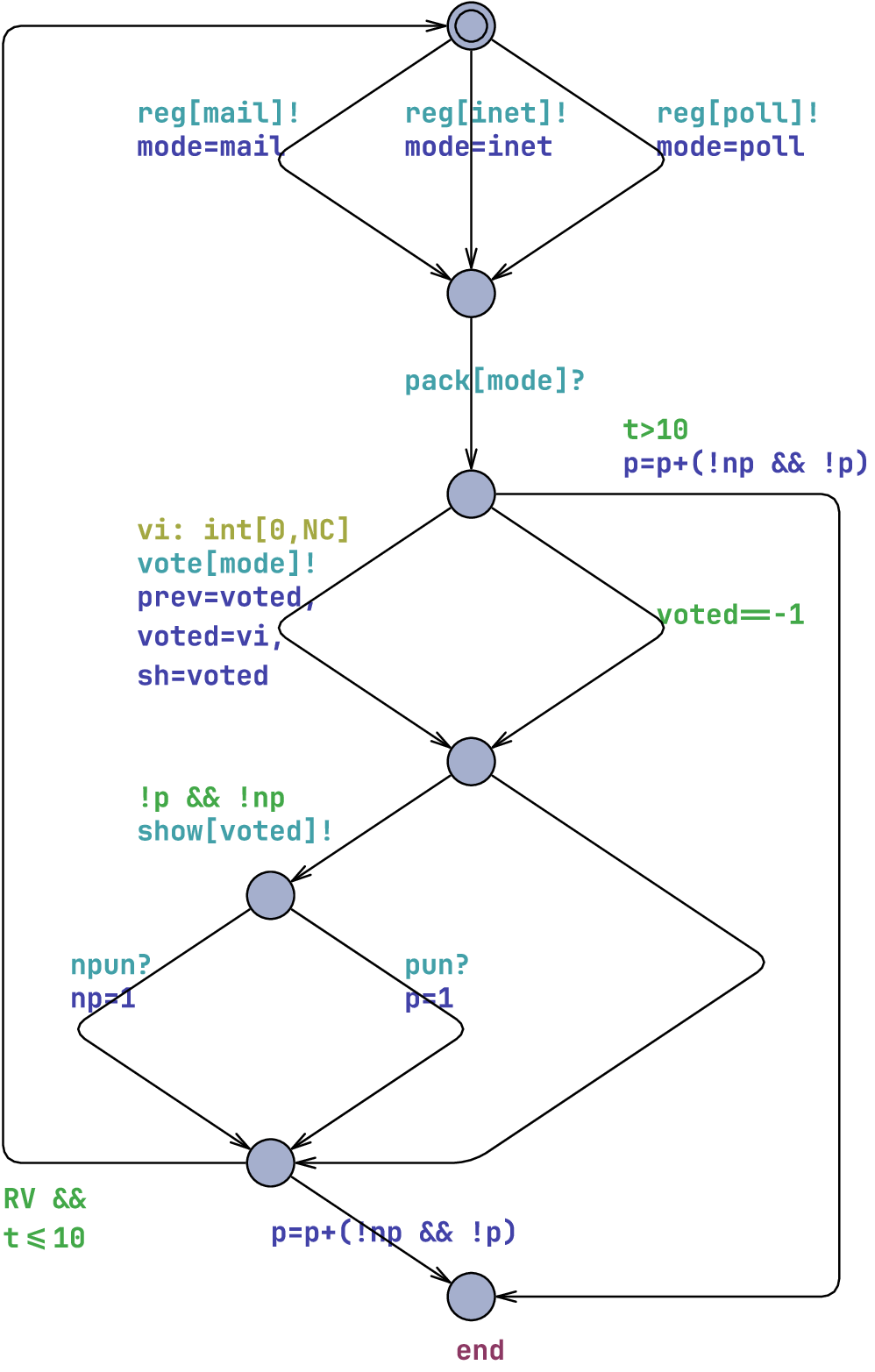}
	\caption{Timed agent graph for the Voter}
	\label{fig:voter-tag}
\end{figure}}
{\setlength{\abovecaptionskip}{0.1ex}%
\setlength{\belowcaptionskip}{0.2ex}%
\begin{figure}[!tbh]
	\centering
	\includegraphics[width=0.95\columnwidth]{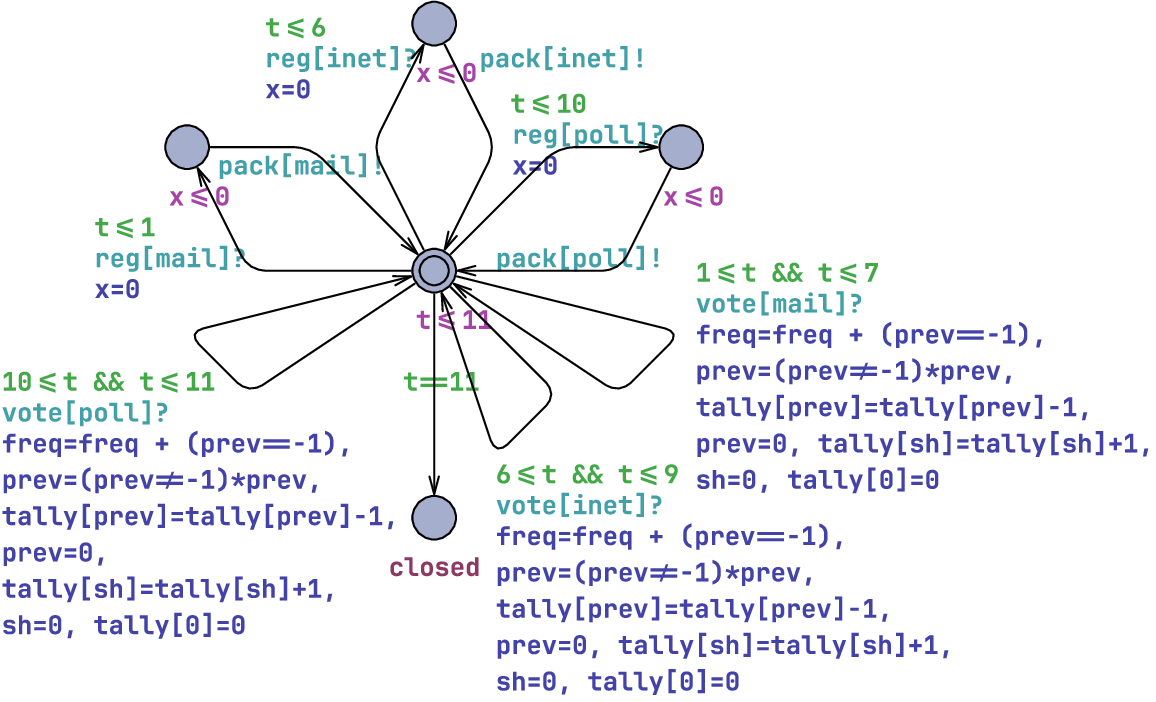}
	\caption{Timed agent graph for the Authority}
	\label{fig:authority-tag}
\end{figure}}
{\setlength{\abovecaptionskip}{0.1ex}%
\setlength{\belowcaptionskip}{0.2ex}%
\begin{figure}[!tbh]
	\centering
	\includegraphics[width=0.74\columnwidth]{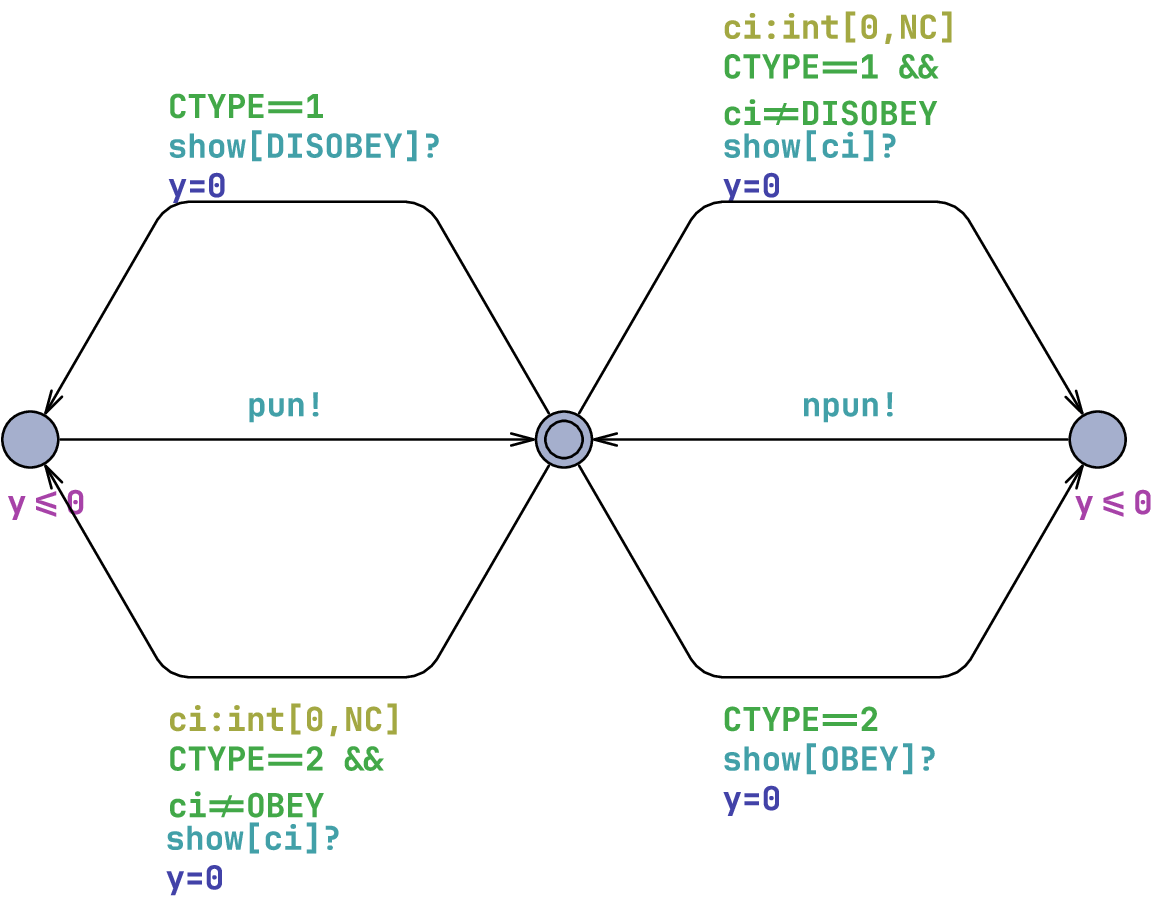}
	\caption{Timed agent graph for the Coercer}
	\label{fig:coercer-tag}
\end{figure}}

\subsection{Multi-Agent Graphs with Clocks}

In this section, we first introduce the specification of an individual agent, and
a set of agents. The associated model is then defined, as well as its behaviour.

\begin{definition}[TAG] %
  \label{def:t-agent-graph}
  A \emph{timed agent graph} (TAG) is a 10-tuple $G=(\Vars, \Loc, l_0, g_0, \Act, \Effect, \Chan, \Clocks, \Inv, \Edges)$, where:
  \begin{itemize}
    \item $\Vars$ is a finite set of variables,
    \item $\Loc$ is a finite set of locations, $l_0\in\Loc$ is the initial location,
    \item $g_0\in\Cond_\Vars$ is the initial condition, %
    s.t. $\restr{\sset{g_0}}{\Vars}$ is a singleton,\footnote{%
      In other words, $ \sset{g_0}\neq \varnothing$ and  $(\eta_1,\eta_2 \in \sset{g_0}) \Rightarrow (\forall_{v\in\Vars}\eta_1(v)=\eta_2(v))$.
    }
    \item $\Act$ is a set of actions, where $\tau\in\Act$ stands for ``do nothing'',
    \item $\Effect: \Act\times\Eval_\Vars\mapsto\Eval_\Vars$ is the effect function, such that $\Effect(\tau,\eta)=\eta$ for any $\eta\in\Eval_\Vars$, 
    \item $\Chan$ is a finite set of asymmetric one-to-one synchronisation channels; by \extended{$\Sync=\set{ch!, ch? \mid ch\in\Chan}\cup\set{-}$}\short{$\Sync$} we denote the set of synchronisation labels of the form ``$ch!$'' and ``$ch?$'' for emitting and receiving on a channel $ch\in \Chan$ respectively, and ``$-$'' for no synchronisation,
    \item $\Clocks$ is a finite set of clocks,
    \item $\Inv: \Loc\mapsto \Constr_\Clocks$ is a location invariant,
    \item $\Edges \subseteq \Loc \times \Cond_\Vars \times  \Constr_\Clocks \times \Sync \times \Act  \times \powerset{\Clocks} \times \Loc$ is a finite set of labelled edges, which define the local transition relation.
  \end{itemize}

  Following the common practice, the notation  $l\xhookrightarrow{g,\clc,\varsigma,\alpha,X}l'$ shall be used as a shorthand for $(l,g,\clc,\varsigma,\alpha,X,l')\in \Edges$. %
  Given a synchronisation label $\varsigma$, its complement denoted by $\overline{\varsigma}$ is defined as $\overline{ch!}=ch?$, $\overline{ch?}=ch!$ and $\overline{\varsigma}=\varsigma$ when $\varsigma=-$.
\end{definition}

Example timed agent graphs for a voting and coercion scenario are shown in~\cref{fig:voter-tag,fig:authority-tag,,,fig:coercer-tag}. %
The scenario is explained in more detail in~\cref{sec:experiments}. 
In order to improve readability, %
the truth valued invariants, %
as well as the edge label components for $g=\top$, $\clc=\top$, $\varsigma=-$, $X=\varnothing$, and $\alpha=\tau$ %
are not explicitly depicted.

\begin{definition}[TMAS Graph]
  \label{def:t-multi-agent-graph}
  A \emph{timed multi-agent system graph} is a multiset%
  \footnote{
    In order to avoid any possible confusion with the ordinary sets, we shall denote
 a multiset container using the ``$\lBrace$'' and ``$\rBrace$'' brackets. %
  } %
  of timed agent graphs $MG=\multiset{G\ag{1}, \ldots, G\ag{\nMG}}$ with distinguished%
  \footnote{Note that the set of shared variables must be explicitly pointed (rather than, for example, being derived from those occurring in two or more agent graphs) due to the possibility of a TMAS graph containing multiple instances of the same agent graph, which in turn may have both shared and non-shared variables.} %
  set of shared variables $\Vars_{\mathit{sh}}$.
  For simplicity, we assume that $MG$ has (some) fixed ordering of its elements.
\end{definition}

\subsection{Models of Timed MAS Graphs}

A \emph{combined TMAS graph} merges the agent graphs in $MG$ into a single agent graph whose nodes represent the possible tuples of locations in $MG$.

\begin{definition}[Combined TMAS Graph]
  \label{def:t-combined-graph}
  Given a TMAS graph $MG=\multiset{G\ag{1}, \ldots, G\ag{\nMG}}$ with the set of shared variables $\Vars_{\mathit{sh}}$,
  where $G\ag{i}=(\Vars\ag{i}, \Loc\ag{i}, l_{0}\ag{i}, g_{0}\ag{i}, \Act\ag{i}, \Effect\ag{i}, \Chan\ag{i}, \Clocks\ag{i}, \Inv\ag{i}, \Edges\ag{i})$, $1\leq i \leq \nMG$, the \emph{combined TMAS graph} of $MG$ is defined as a timed agent graph $G_{MG}=(\Vars, \Loc, l_0, g_0, \Act, \Effect, \Chan, \Clocks, \Inv, \Edges)$, where ${\Chan=\varnothing}$, 
  $\Vars = \bigcup_{i=1}^{\nMG} \Vars\ag{i}$, $\Loc=\prod_{i=1}^{\nMG} \Loc\ag{i}$, $l_0=(l_{0}\ag{1},\ldots,l_{0}\ag{\nMG})$, $\Clocks=\bigcup_{i=1}^{\nMG}\Clocks\ag{i}$, and 
  $g_0=(g_{0}\ag{1}\wedge\ldots\wedge g_{0}\ag{\nMG})$.
  The location invariant is defined as $\Inv(l\ag{1},\ldots,l\ag{\nMG})=\Inv\ag{1}(l\ag{1})\wedge\ldots\wedge\Inv\ag{\nMG}(l\ag{\nMG})$ and the set of actions $\Act = \set{ \alpha\ag{u_1}\updcomp\ldots\updcomp\alpha\ag{u_k}\mid \alpha\ag{i}\in\Act\ag{i}, i\in \{u_1,\ldots, u_k\}\subseteq \set{1,\ldots,\nMG}}$,
  The effect function $\Effect:\Act\times\Eval_\Vars\mapsto \Eval_\Vars$ is defined by:
  \[\small\Effect(\alpha,\eta) =
  \begin{cases}
    \eta[\Vars\ag{i}=\Effect\ag{i}(\alpha,\eta_{|\Vars\ag{i}})(\Vars\ag{i})] &\text{if } \alpha\in \Act\ag{i}\\
    \Effect(\alpha\ag{i},\Effect(\alpha\ag{j},\eta)) &\text{if } \alpha = \alpha\ag{i}\updcomp\alpha\ag{j}
  \end{cases}\]
  where $\eta[X=Y]$ denotes the evaluation $\eta'$, such that $\restr{\eta'}{X}=\set{Y}$ and $\restr{\eta'}{\Vars\setminus X}=\restr{\eta}{\Vars\setminus X}$. 

  \noindent
  Given a pair of agents $G\ag{i}$ and $G\ag{j}$ of distinct indices $i\neq j$ with $(l_1,g_1,\clc_1,\varsigma,\alpha_1,X_1,l_{1}')\in \Edges\ag{i}$, $(l_2,g_2,\clc_2,\overline{\varsigma},\alpha_2,X_2,l_{2}')\in \Edges\ag{j}$, the set of labelled edges $\Edges$ is obtained inductively by applying the  rules:\\[6pt]
  \noindent
\begin{tabularx}{\linewidth}{@{}l r}
	$\begin{prooftree}
			l_{1} \lhook\joinrel\xrightarrow{g_1,\clc_1,ch!,\alpha_1,X_1} l_{1}',\ l_{2} \lhook\joinrel\xrightarrow{g_2,\clc_2,ch?,\alpha_2,X_2} l_{2}'
			\justifies
			l_{1},l_{2}\lhook\joinrel\xrightarrow{g_1\wedge g_1, \clc_1\wedge\clc_2, -, \alpha_2\updcomp \alpha_1, X_1\cup X_2 }l_{1}',l_{2}'
  \end{prooftree}$	
	&%
	$\begin{prooftree}
		l_{1} \lhook\joinrel\xrightarrow{g_1,\clc_1,-,\alpha_1,X_1} l_{1}'
		\justifies
		l_{1},l_{2} \lhook\joinrel\xrightarrow{g_1,\clc_1,-,\alpha_1,X_1} l_{1}',l_{2}
  \end{prooftree}$
	\\
	\\
	$\begin{prooftree}
    l_{1} \lhook\joinrel\xrightarrow{g_1,\clc_1,ch?,\alpha_1,X_1} l_{1}',\ l_{2} \lhook\joinrel\xrightarrow{g_2,\clc_2,ch!,\alpha_2,X_2} l_{2}'
    \justifies
    l_{1},l_{2}\lhook\joinrel\xrightarrow{g_1\wedge g_1, \clc_1\wedge\clc_2, -, \alpha_1\updcomp \alpha_2, X_1\cup X_2 }l_{1}',l_{2}'
  \end{prooftree}$	
	&%
	$\begin{prooftree}
		l_{2} \lhook\joinrel\xrightarrow{g_2,\clc_2,-,\alpha_2,X_2} l_{2}'
		\justifies
		l_{1},l_{2} \lhook\joinrel\xrightarrow{g_2,\clc_2,-,\alpha_2,X_2} l_{1},l_{2}'
  \end{prooftree}$
\end{tabularx}
\end{definition}

The \emph{model} further unfolds the combined TMAS graphs by explicitly representing the reachable valuations of model variables.

\begin{definition}[Model]
  \label{def:t-model}
  The \emph{model} of a timed agent graph $G$ over $AP$ is a 5-tuple $\mathcal{M}(G)=(\States, \sini, \longrightarrow, AP, L)$, where:
  \begin{itemize}
    \item $\States = \Loc\times \Eval_{\Vars} \times \Rplus^{n_{\Clocks}}$ is a set of global states,
    \item $\sini=(l_0,\eta_0,\upsilon_0)\in\States$, such that  $\eta_0\in \sset{g_0}$ and $\forall_{0\leq i\leq n_{\Clocks}}\upsilon_0(x_i)=0$, is an initial global state,
    \item $\longrightarrow\subseteq \States\times\States$ is the transition relation composed of:\\
    -- \emph{delay-transitions}: $(l,\eta,\upsilon)\xlongrightarrow{{\delta}}(l,\eta,\upsilon+\delta)$ for $\delta\in \Rplus$ with $\upsilon,\upsilon+\delta \in \sset{\Inv(l)}$,\\
    -- \emph{action-transitions}: $(l,\eta,\upsilon)\xlongrightarrow{\smash{\alpha}}(l',\eta',\upsilon')$ for $l\xhookrightarrow{\smash{g,\clc,-,\alpha,X}}l'$ with $\upsilon\in\sset{\Inv(l)}$, $\upsilon'\in\sset{\Inv(l')}$,  $\eta\in\sset{g}$, $\upsilon\in\sset{\clc}$, $\eta'=\Effect(\alpha,\eta)$ and $\upsilon'=\upsilon[X=0]$, %
    \item $AP\subseteq (\Cond_{\Vars}\cup\Loc)$ is a finite set of atomic propositions,
    \item $L: \States \mapsto \powerset{AP}$ is a labelling function, such that\\$L(l,\eta,\upsilon)\subseteq(\set{g\in \Cond_{\Vars} \mid \eta\in\sset{g}}\cup\set{l})$.
  \end{itemize}
  The model $\mathcal{M}$ of a TMAS graph $MG$ is given by the model of its combined TMAS graph, i.e. $\mathcal{M}(MG)=\mathcal{M}(G_{MG})$.
\end{definition}

We now formally define paths of the model.
For this paper, we consider only \emph{progressive}%
\footnote{Also called \emph{time-divergent}.} %
paths that are free from the timelocks and deadlocks.
In general, this is the property that valid models of the system should have.
Furthermore, such a condition can be checked both statically and dynamically (see
e.g.,~\cite{Penczek06advances}).

\begin{definition}
  A \emph{path} from the state $s_0\in\States$ of the model $M$ is an infinite sequence of states $\pi=s_0 s_1 s_2 \ldots$, such that $s_i\in\States$, $s_{2i}\xlongrightarrow{\delta_i} s_{2i+1}\xlongrightarrow{\alpha_i}s_{2i+2}$, where $\delta_i\in\Rplus$, $\alpha_i\in \Act$, for every $i \geq 0$ and $\sum_{i\in\mathbb{N}}\delta_i=\infty$.
  Given $i\geq 0$, by $\pii{i}=s_i$ and $\pii{i:}=s_i s_{i+1}\ldots$ we denote $i$-th state and $i$-th suffix of $\pi$ respectively.
  An \emph{initial path} of a model $M$ is a path $\pi$ that starts with the initial state
  $\sini$, that is $\pii{0}=\sini$.
  The set of all paths of $M$ is denoted by $\Paths_M$, the set of all paths starting in $s\in\States$ \extended{is denoted }by $\Paths_M(s)$.\footnote{%
    When a model $M$ is clear from the context, the subscript $M$ is omitted.
  }

  A state $s\in\States$ is \emph{reachable} in $M$ iff there exists an initial path $\pi\in\Paths_M$, such that $\pi[i]=s$ for some $0\leq i < \infty$.
  The set of all reachable states in $M$ is denoted by $\Reach(M)$.

  An evaluation $\eta\in\Eval_\Vars$ is \emph{reachable} at $l\in\Loc$ iff there is reachable state of the form $(l,\eta,\upsilon)\in\Reach(M)$ for some $\upsilon\in\Rplus^{n_{\Clocks}}$.

  A \emph{local domain} is a function $d:\Loc\mapsto \powerset{\Eval_{\Vars}}$ that maps every location $l\in\Loc$ to the set of its reachable evaluations, that is $d(l)=\set{\eta\in\Eval_{\Vars} \mid (l,\eta,\upsilon)\in\Reach(M)}$.
\end{definition}

\subsection{Logical Reasoning about TMAS Graphs}

We shall now define the branching-time (timed) logic $\TCTLs$~\cite{bouyer2018model} that generalizes $\TCTL$~\cite{alur1993model}, $\CTL$~\cite{Clarke81ctl}.

For a set of atomic propositions $AP$, the syntax of \TCTLs state formulae $\varphi$ and path formulae $\psi$ is given by the following grammar:
\begin{gather*}
  \varphi ::=\ p\ \mid\ \neg \varphi \ \mid\ \varphi \vee \varphi \ \mid\ \Apath \psi \ \mid\ \Epath \psi, \\
  \psi ::=\ \varphi\ \mid\ \neg \psi\ \mid\ \psi \vee \psi \ \mid\ \psi \Untili \psi,
\end{gather*}
\noindent%
where $p\in AP$, $I$ is an interval in $\Rplus$ with integer bounds of the form $[a,b]$, $(a,b]$, $[a,c)$, $(a,c)$ for $a,b\in\mathbb{N}$, $c\in\mathbb{N}\cup\set{\infty}$. %
The temporal operator $\Until$ stands for ``until'', the path quantifiers $\Apath$ and $\Epath$ stand for ``for all paths'' and ``exists a path'' respectively.
Intuitively, an interval $I$ constrains the modal operator it subscribes to. %
Boolean connectives and additional constrained temporal operators ``sometime'' (denoted by $\Sometm_I$) and ``always'' (denoted by $\Always_I$) can be derived as usual.
In particular, $\Sometm_I \varphi \equiv \top \Untili \varphi$, $ \Always_I \varphi \equiv \neg\Sometm_I \neg \varphi$.
We will sometimes omit the subscript for $I=[0,\infty)$, writing $\Until$ as a shorthand for $\Until_{[0,\infty)}$.

From the above, we can define other important logics:
\begin{description}[leftmargin=12mm]
  \item[\TCTL\hspace{2mm}] restriction of $\TCTLs$, where every occurrence of temporal operators is always preceded with a path quantifier,
  \item[\TACTL] restriction of $\TCTL$, where negation can only be applied to propositions
  and only the $\Apath$ path operator is used,
  \item[\CTLs\hspace{2mm}] obtained from \TCTLs with only trivial subscript intervals $I=[0,\infty)$, adding modal operator ``next'' to the syntax, 
  \item[\ACTLs] restriction of $\CTLs$, where negation can only be applied to propositions
  and only the $\Apath$ path operator is used.
\end{description}

Let $\pi=s^{\ }_0 s^{\delta_0}_0 s^{\ }_1 s^{\delta_1}_1\ldots$ be a path of the model $M$, s.t. $s^{\ }_i\xlongrightarrow{\delta_i}s^{\delta_i}_i$ and $s^{\delta_i}_i\xlongrightarrow{\alpha_i}s^{\ }_{i+1}$ for $i\geq 0$, where 
$s^{\delta}_i=(l_i,\eta_i,\upsilon_i+\delta)$ for $\delta\in\Rplus$ and $s^{\ }_i=s^{0}_i$, 
and let $\pi(i,\delta)=s^{\delta_{\,}}_i s^{\delta_i}_i s^{\ }_{i+1} s^{\delta_{i+1}}_{i+1}\ldots$  denote the suffix of $\pi$ for $i\geq 0$ and $\delta\leq \delta_i$, such that $s^{\delta_{\,}}_i\xlongrightarrow{\delta_i-\delta}s^{\delta_i}_i$.
The semantics of \TCTLs is as follows:\footnote{The omitted clauses for the Boolean connectives are immediate.} %
\begin{alignat*}{3}
  & M,s\models p                             &  & \quad\text{iff }p\in L(s)                                                         \\
  & M,s\models \Apath\psi                 &  & \quad\text{iff } M,\pi\models\psi\text{, for all }\pi\in\textit{Paths}(s)       \\
  & M,s\models \Epath\psi                 &  & \quad\text{iff } M,\pi\models\psi\text{, for some }\pi\in\textit{Paths}(s)       \\[\semanticsnewlinebreak]
  & M,\pi\models \varphi                     &  & \quad\text{iff } M,\pii{0}\models\varphi \\
  & M,\pi\models \psi_1\Untili\psi_2   &  & \quad\text{iff } \exists i\geq 0\ldotp\exists \delta\leq \delta_i\ldotp\left(
                                            (\textstyle\sum_{j < i}\delta_j + \delta)\in I\right) \text{, and }\\
  &                                    &  & \quad\hspace*{\widthof{\text{iff }}}  
                                            M,\pi(i,\delta) \models \psi_2\text{, and }\\
  &                                    &  & \quad\hspace*{\widthof{\text{iff }}}
                                            \forall \delta'<\delta\ldotp M,\pi(i,\delta') \models \psi_1\text{, and }\\ 
  &                                    &  & \quad\hspace*{\widthof{\text{iff }}}
                                            \forall {j < i}\ldotp\forall {\delta' \leq \delta_j}\ldotp M,\pi(j,\delta')\models\psi_1.
\end{alignat*}
\noindent 
The model $M$ satisfies the \TCTLs (state) formula $\varphi$ (denoted by $M\models \varphi$) iff $M,\sini\models \varphi$.

\section{Simulation for TMAS Graphs}

We now propose a notion of \emph{simulation} between timed MAS graphs, that is later used to establish correctness of our abstraction scheme.

\begin{definition}
  \label{df:model-sim}
  Let $M_i=(\States_i, \sini_i, \longrightarrow_i, AP_i, L_i)$, $i=1,2$.
  A model $M_2$ \emph{simulates} model $M_1$ over $AP\subseteq AP_1\cap AP_2$ (denoted $M_1\precsim_{AP} M_2$) if there exists a \emph{simulation} relation $\mathcal{R}\subseteq \States_1\times \States_2$ such that:
  \begin{itemize}
    \item[(i)] $(\sini_1,\sini_2)\in\mathcal{R}$, and
    \item[(ii)] for all $(s_1,s_2)\in\mathcal{R}$:
    \begin{itemize}
      \item[(a)] $L_1(s_1)\cap AP = L_2(s_2)\cap AP$, and  %
      \item[(b)] if $s_1\longrightarrow_1 s_1'$, then $\exists{s'_2\in\States_2}$ s.t. $ s_2\longrightarrow_2 s_2'$ and $(s_1',s_2')\in\mathcal{R}$.
    \end{itemize}
  \end{itemize}

  If additionally $(s_1,s_2)\in \mathcal{R} \Rightarrow (\upsilon_1=\upsilon_2)$, where $s_i=(l_i,\eta_i,\upsilon_i)$ for $i=1,2$, then $\mathcal{R}$ is called the \emph{timed simulation} relation~\cite{bouyer2018model}.%
\end{definition}

\begin{theorem}
  \label{thm:tactl-pres}
  If there is timed simulation $\mathcal{R}\subseteq \States_1\times\States_2$ over $AP\subseteq AP_1\cap AP_2$, then for any formula ${\varphi\in\TACTLs}$ over $AP$: %
  $$M_2\models \varphi \qquad\text{implies}\qquad M_1\models \varphi.$$ 
\end{theorem}

\begin{proof}[Proof (Sketch)]
 The results showing that simulation preserves \ACTL and \ACTLs are well established~\cite{alur1993model,penczek2001abstractions}, this can be proven using structural induction on formula (cf.~\cite{Penczek06advances,baier2008principles}). 
  Almost the same line of reasoning can be applied to timed simulation and $\TACTLs$. 
  Here, we show that for a more interesting case, when $\varphi$ is of the form $\Apath \psi_1\Untili\psi_2$; the remaining cases are shown analogously as in~\ACTLs.
  
  Let $\mathcal{R}\subseteq\States_1\times\States_2$ be a timed simulation for $(M_1,M_2)$, and $\pi=s^{\ }_{i,0} s^{\delta_{i,0}}_{i,0}  s^{\ }_{i,1}  s^{\delta_{i,1}}_{i,1}\ldots$ denote a path of the model $M_i$ for $i=1,2$, such that $s^{\ }_{i,j}\xlongrightarrow{\delta_{i,j}}s^{\delta_{i,j}}_{i,j}$ and $s^{\delta_{i,j}}_{i,j}\xlongrightarrow{\alpha_{i,j}}s^{\ }_{i,j+1}$, 
  where $s^{\ }_{i,j}=(l_{i,j},\eta_{i,j},\upsilon_{i,j})$ and $s^{\delta}_{i,j}=(l_{i,j},\eta_{i,j},\upsilon_{i,j}+\delta)$, and let $\pi_{i}(j,\delta)=s_{i,j}^{\delta_{{\,}_{\,}}} s^{\delta_{i,j}}_{i,j} s^{\ }_{i,j+1} s^{\delta_{i,j+1}}_{i,j+1}\ldots$ denote the suffix of $\pi$ for $j\geq 0$ and $\delta\leq \delta_{i,j}$, s.t. $s^{\delta_{\,}}_{i,j}\xlongrightarrow{\delta_{i,j}-\delta}s^{\delta_{i,j}}_{i,j}$.

  Let $\pi_1\in \Paths_{M_1}(\sini_1)$ be an arbitrarily chosen path. 
  From~\cref{df:model-sim} we construct a matching to $\pi_1$ path $\pi_2\in \Paths_{M_2}(\sini_2)$, s.t. 
  $(s^{\ }_{1,j}, s^{\ }_{2,j}), (s^{\delta_{1,j} }_{1,j}, s^{\delta_{2,j} }_{2,j})\in\mathcal{R}$ for all $j\geq 0$. 
  Since $\mathcal{R}$ is a timed simulation, it follows that $\upsilon_{1,j}=\upsilon_{2,j}$ and $\delta_{1,j}=\delta_{2,j}$ for all $j\geq 0$. 
  From $M_2,\sini_2 \models \Apath \psi_1\Untili\psi_2$, 
  we know that for $\pi_2\in\Paths_{M_2}(\sini_2)$ the following holds: $\exists j\geq 0 \ldotp \exists \delta\leq \delta_{2,j}\ldotp (\textstyle\sum_{k < j}\delta_{2,k} + \delta)\in I$, and $M_2,\pi_2(j,\delta)\models \psi_2$, and $\forall \delta'< \delta \ldotp M_2,\pi_2(j,\delta')\models \psi_1$, and $\forall k< j \ldotp \forall \delta'< \delta_{2,k} \ldotp M_2,\pi_2(k,\delta')\models \psi_1$; consequently, for $\pi_1$ it follows that: 
  $\delta\leq \delta_{1,j}$ and $(\textstyle\sum_{k < j}\delta_{1,k} + \delta)\in I$, and by induction that $M_1,\pi_1(j,\delta)\models \psi_2$, and $\forall \delta'< \delta \ldotp M_1,\pi_1(j,\delta')\models \psi_1$, and $\forall k< j \ldotp \forall \delta'< \delta_{1,k} \ldotp M_1,\pi_1(k,\delta')\models \psi_1$.  
  Since $\pi_1$ was chosen arbitrarily, the same reasoning can be applied to any path in $\Paths_{M_1}(\sini_1)$; hence we have $M_1,\sini_1 \models \Apath \psi_1\Untili\psi_2$.

\end{proof}

Given a timed agent graph $G$, its time-insensitive variant $G\uext$ is a timed agent graph that is an almost identical copy of $G$ except that the set of clocks  $\Clocks$ is set to $\varnothing$.
Essentially, it also means that in $G\uext$, for any $l \in\Loc$ an invariant function is such that $\Inv(l)=\top$, and for any $(l,g,\clc,\varsigma,\alpha,X,l')\in \Edges$, $\clc$ and $X$ are replaced with $\top$ and  $\varnothing$ respectively. 

For a TMAS graph $MG=\multiset{G\ag{1}, \ldots, G\ag{\nMG}}$, a time-insensitive variant  $MG\uext$ is given through the time-insensitive variant of the agent graphs composing it, that is $MG\uext=\multiset{G\ag{1}\uext, \ldots, G\ag{\nMG}\uext}$

\begin{lemma}
  Any evaluation $\eta\in\Eval_\Vars$ reachable at $l\in\Loc$ for an agent graph $G$ must also be reachable at $l$ for $G\uext$.
\end{lemma}

\begin{proof}
  Follows directly from~\cref{def:t-model}.
\end{proof}

\begin{corollary}
  \label{over-app-corr}
Given a pair of agent graphs $G_1$ and $G_2$ with their local domain functions $d_1$ and $d_2$, such that $d_i:\Loc_i\mapsto \powerset{\Eval_{\Vars_i}}$ maps every location $l\in\Loc_i$ to the set of its reachable evaluations, i.e. $d_i(l)=\set{\eta\in\Eval_{\Vars_i} \mid (l,\eta,\upsilon)\in\Reach(\mathcal{M}(G_i))}$, for $i=1,2$,
in case $G_2={G_1}\uext$ is a time-insensitive variant of $G_1$, then its local domain $d_2$ is an over-approximation of the local domain $d_1$, that is $d_1(l)\subseteq d_2(l)$ for all $l\in\Loc_1=\Loc_2$.
\end{corollary}

\begin{algorithm}[t]
    \DontPrintSemicolon
    in $G_{MG\uext}$ compute an over-approx. of local domain $d$ for $\Argsd{\mathcal{F}}$\;
    \ForEach{\text{timed agent graph }$G\ag{i}\in MG$}{
        compute an abstract timed agent graph $\mayabstr_{\mathcal{F}}(G\ag{i})$ w.r.t. $d_i=\set{l\ag{i}\mapsto \bigcup_{l\in\Loc\ag{1}\times\ldots\times\set{l\ag{i}}\times\ldots\times\Loc\ag{\nMG}}d(l)}$\;
    }
    \Return $\mayabstr_{\mathcal{F}}(MG)=\multiset{\mayabstr_{\mathcal{F}}(G\ag{1}),\ldots,\mayabstr_{\mathcal{F}}(G\ag{\nMG})}$
    \caption{Abstraction of TMAS graph $MG$ wrt $\mathcal{F}$}
\label{alg:abstraction-idea}
\end{algorithm}

\begin{algorithm}[!t]
	\DontPrintSemicolon
	\SetKw{KwGoTo}{go to}
	\SetKwFor{While}{while}{}{end while}%
	\nonl\myproc{\overApproxLocalDomain{$G=G_{MG}, W=\Argsd(f)$}}{
		\ForEach{ $l \in \Loc$}{
			\label{alg:approx-ld-init1}
			$l.d := \varnothing$\; %
			$l.p:=\varnothing$\; %
			$l.\mathit{color}:=\textit{white}$\;
		}
		$l_0.d := \set{\eta(W) \mid \eta\in \sset{g_0}}$\; %
		\label{alg:approx-ld-init2}
		$Q:=\varnothing$\;%
		\label{alg:approx-ld-q1}
		$\enqueue(Q,l_0)$\;
		\label{alg:approx-ld-q2}
		\label{alg:approx-ld-mainloop1}
		\While{$Q\neq \varnothing$}{
			$l:=\extractMax(Q)$\;%
			$\visitLoc(l,W)$\;
			\label{alg:approx-ld-enqueue1}
			\If{$l.\mathit{color}\neq\textit{black}$}{\tikzmark{top1}
				\ForEach{$l'\in\mathit{Succ}(l)$}{
					$Q:=\enqueue(Q,l')$\;
					$l'.p:=l'.p\cup\set{l}$\tikzmark{bottom1}\; 
				}
				$l.\mathit{color}:=\textit{black}$\;
				\label{alg:approx-ld-enqueue2}
			}
			\label{alg:approx-ld-mainloop2}
		}
		\KwRet{$\set{l\mapsto \eta \mid l\in\Loc, \eta\in\Eval_{\Vars}, \eta(W)\in d.l}$}\;
		\label{alg:approx-ld-return}
	}
	\smallskip
	\nonl\myproc{\visitLoc{$l, W$}}{
		$\kappa:=l.d$\tikzmark{top2}\;
		\label{alg:approx-ld-procincedges1}
		\ForEach{$l'\in l.p,\; l'\smash{\xhookrightarrow{g,\clc,-,\alpha,X}}l$}{ %
			$l.d:=l.d \cup \procEdgeLabels(l',g,\clc,-,\alpha,X,l,W)$
			\label{alg:approx-ld-aux1}
		}
		$l.p=\varnothing$\;
		\label{alg:approx-ld-pireset}
		\If{$\kappa \neq l.d$}{
			$l.\mathit{color}:=\textit{grey}$\tikzmark{bottom2}\;
			\label{alg:approx-ld-procincedges2}
		}
		$\lambda:=l.d$\tikzmark{top3}\;\label{proc-self-loops}\label{alg:approx-ld-procloops1}
		\ForEach{$l\smash{\xhookrightarrow{g,\clc,-,\alpha,X}} l$}{
			$l.d:=l.d\cup \procEdgeLabels(l,g,\clc,-,\alpha,X,l,W)$\quad\tikzmark{right1}
			\label{alg:approx-ld-aux2}
		}
		
		\If{$\lambda\neq l.d$}{
			$l.\mathit{color}:=\textit{grey}$\;
			\KwGoTo \ref*{proc-self-loops}
			\label{alg:approx-ld-procloops2}\tikzmark{bottom3}
		}
	}
	
	\smallskip
	\nonl\myproc{\procEdgeLabels{$l,g,\clc,-,\alpha,X,l',W$}}{
		$H_{\textit{pre}} :=  \set{\eta\in \sset{g} \mid \eta(W)\in l.d}$\;
    $H_{\textit{post}}:=\set{ \Effect(\alpha,\eta) \mid \eta\in H_{\textit{pre}}}$\;
		\KwRet{$\set{\eta(W) \mid \eta\in H_{\textit{post}}}$}\AddNoteSmall{top1}{bottom1}{right1}{enqueue immediate-neighbours}\AddNoteSmall{top2}{bottom2}{right1}{process incoming edges}\AddNoteSmall{top3}{bottom3}{right1}{process self-loops}
	}	
	\caption{\mbox{An over-approx. of the local domain for $W\subseteq\Vars$}}
	\label{alg:variable-approx}
\end{algorithm}

 \begin{algorithm}[!t]
  \DontPrintSemicolon
  \caption{Abstraction procedure}
  \label{alg:general-abstraction}
  \nonl\myproc{\generalVarAbstraction{$G$, $d$, $f$, $\Scope$}}{
    $Z_0:=(f(\restr{\sset{g_0}}{W}))(Z)$\label{alg:gen-abstr-l0}\;
    $g_0:=g_0\wedge(Z{=}Z_0)$\label{alg:gen-abstr-l1}\;
    $\eta_0\in \sset{g_0}$\;
    $\Edges_a:=\varnothing$\;
    \ForEach{$e := \smash{l\xhookrightarrow{g,\clc,\varsigma,\alpha,X}l'}$}{
    \label{alg:gen-abstr-main-for1}
      \eIf{$\set{l,l'}\cap\Scope=\varnothing$\tikzmark{t1}}{
        \label{alg:gen-abstr-inner-if1}
        $\Edges_a := \Edges_a \cup \set{e}$\tikzmark{b1}\;
      }{
        \label{alg:gen-abstr-inner-else1}      
        \ForEach{$\eta\in d(l)$}{
          \label{alg:gen-abstr-inner-for1}
          $\alpha' := \alpha$\hspace{4.7em}\tikzmark{r1}\;
          \If{$l\in \Scope$}{\label{alg:gen-abstr-pre-if1}
            $\alpha' := (W:=\eta(W)).\alpha'$\label{alg:gen-abstr-pre-if2}\; %
            $\alpha' := (Z:=(\eta_0(Z))).\alpha'$  %
            \label{alg:gen-abstr-pre-if3}
          }
          \If{$l'\in \Scope$}{\label{alg:gen-abstr-post-if1}
            $\alpha' := \alpha'.(Z:=(f(\restr{\eta}{W}))(Z))$\label{alg:gen-abstr-post-if2}\; %
            $\alpha' := \alpha'.(W:=\eta_0(W))$  %
            \label{alg:gen-abstr-post-if3}
          }
          $g':=g[W=\eta(W)]$\;
          $\Edges_a := \Edges_a \cup \set{(l,g',\clc,\varsigma,\alpha',X,l')}$
          \label{alg:gen-abstr-main-for2}
        }
      }
    }
    $\Edges:=\Edges_a$\;
    $\Vars:=\Vars\cup Z$\;
    \KwRet{$G$}\AddNoteSmall{t1}{b1}{r1}{\texttt{out-of-scope edge}}
  }
\end{algorithm}%

\section{Variable Abstraction for Timed MAS}
\label{sec:abstraction}

This section presents the abstraction method that is defined for the modular representation as TMAS graph and intended to reduce the state space of the induced abstract model.

\subsection{Definition}

As in~\cite{Jamroga23easyabstract}, the abstraction process is composed of two subroutines:
computation of the local domain approximation followed by abstract model generation.
\footnote{Essentially, the former subroutine might be skipped, when the required approximation is provided by the user.}
The intuition behind this is informally presented in~\cref{alg:abstraction-idea}.

Formally, the abstraction $\mathcal{A}^\mathit{may}_\mathcal{F}$ for TMAS graph $MG$ is specified by the set of pairs $\mathcal{F}=\set{(f_1, \Scope_1),\ldots,(f_{n_{\mathcal{F}}},\Scope_{n_{\mathcal{F}}})}$, where $f_i: \Eval_{W_i}\mapsto \Eval_{Z_i}$ is an abstract mapping function, such that $W_i\subseteq \Vars$, $Z_i\cap \Vars = \varnothing$ and $i\neq j \Rightarrow Z_i\cap Z_j=\varnothing$, and $\Scope_i\subseteq \Loc$ is the effective scope, for any $1\leq i, j\leq n_{\mathcal{F}}$.
Intuitively, the scope enables applying a finer-grain abstractions, only for the certain fragment of the system.
We denote the domain and range of mapping function $f_i$ by $\Argsd(f_i)=W_i$ and $\Argsr(f_i)=Z_i$.
For simplicity, in the sequel we restrict the presentation to the case of singleton $\mathcal{F}=\set{(f,\Scope)}$ --- the changes needed for the general case are merely technical and shall become apparent afterwards.

Intuitively, the abstraction obtains significant improvements under the following conditions: the verified property is *independent* from the variables being removed (so that the abstraction is conclusive), and the removed variables are largely *independent* from those being kept (so that it significantly reduces the state space). 
The *independence* is of semantic nature, and we see no easy way to automatically select such variables. At this stage, it seems best to follow a domain expert advice. 

As follows from the discussion in the extended version of~\cite{Jamroga23easyabstract},
lower-approximation of local domain is usually of little use in practice, as it can only map a location with a singleton or an empty set.
Therefore, here we focus only on the may-abstraction procedure for the TMAS graph
that is based on the over-approximation of local domain. 

\para{Local domain approximation.}
The $\overApproxLocalDomain$ from \cref{alg:variable-approx} on the input takes a timed agent $G$ (usually being a combined TMAS graph) and a subset of its variables $W\subseteq\Vars$, and then traverses the locations of $G$ in a priority-BFS manner computing for each location the set of its possibly reachable evaluations of $W$.
With each visit to a location, an algorithm attempts to refine the approximation,
until stability (in terms of set-inclusion) is obtained.
Starting from the initial one, every location must be visited at least once, and shall be re-visited again whenever any of its predecessors gets their approximation updated.
As the number of locations and edges, and cardinality of the variables domains are all finite, the procedure is guaranteed to terminate.

Let $\dmax:=\textrm{max}\set{|\Dom(v)| \mid v\in W}$, $k=|\Vars|$, $r=|W|$, $n=|\Loc|$, $m=|\Edges|$.
The initialization loop on lines 1--4 takes $O(n)$ and line 5 $O(r)$. 
With generic heap implementation of priority queue, lines 6 and 7 run in $\Theta(\log{n})$ and $\Theta(1)$ respectively. 
The loop on lines 8--15 repeats at most $\dmax^r$ times for each of $n$ locations until stable approximation is obtained. 
Line 9 is in $O(\log{n})$, lines 11--15 in $O(n^2)$. 
Upon visiting all $n$ locations $\dmax^r$ times, $\visitLoc$ calls to $\procEdgeLabels$ at most $m\dmax^r$ times.
Hence, given the set-union is in $O(\dmax^r)$ and $\Effect$ can be computed in $O(1)$, the running time of the whole $\overApproxLocalDomain$ is $O(m\dmax^r(k+\dmax^r) + n^3\dmax^r)$. 
It needs $O(n)$ space for priority queue, $O(m\dmax^r)$ for auxiliary look-up tables (e.g., for the pre-computation of $\sset{g}$) and $O(n\dmax^r + n^2)$ for storing locations with their attributes.
Hence, $\overApproxLocalDomain$ is in $O(\dmax^r(n+m)+n^2)$ space.

\para{Abstraction generation.} 
The abstraction computation procedure is described in \cref{alg:general-abstraction}. 
In the main loop (lines~\ref{alg:gen-abstr-main-for1}--\ref{alg:gen-abstr-main-for2}) the edges of timed agent graph $G$ are transformed according to $\mathcal{F}=\set{(f,\Scope)}$ as follows: 
\begin{itemize}
\item the edges entering or within $\Scope$ have their actions appended with (1) update of the target variable $Z$ and (2) update which sets the values of the source variables $W$ to their defaults (resetting those),
\item the edges leaving or within $\Scope$ have actions prepended with (1) update of source variables $W$ (a temporarily one to be assumed  for the original action) and (2) update which resets the values of the target variable $Z$.
\end{itemize}

Note that due to introduction of a scope, the variables in $X$ cannot be genuinely removed for a proper subset of locations --- 
instead, their evaluation are fixed to some constant value within the states, where the corresponding location label falls under the scope.

The lines 1--4 run in $O(r)$, the outer loop on lines 5--17 repeats exactly $m$ times, the ``if-else'' condition check involves set operation and runs in $O(n)$. 
For the worst-case analysis we shall assume that it always proceeds with ``else'' block. 
The inner loop on the lines 8--17 repeats at most $\dmax^r$ times. 
Assuming that operations on lines 11--12, 14--15 are in $O(1)$ time,\footnote{In general, having a complex or lengthy edge labels is considered a bad practice, which significantly degrades the model readability; realistically, the lengths of edge labels can be assumed to be relatively small.} %
we can conclude that $\generalVarAbstraction$ runs in $O(r+m(n + m\dmax^r))$. 
It requires at most $O(m\dmax^r)$ for storing new edges, and $O(n)$ for $\mathcal{F}$\footnote{The latter would be $O(nn_{\mathcal{F}})$ when $\mathcal{F}$ is not a singleton.}. 
Hence, space complexity of $\generalVarAbstraction$ is in $O(m\dmax^r+n)$.

\subsection{Correctness}

\begin{theorem}
  Let $M_1=\mathcal{M}(MG)$, $M_2=\mathcal{M}(\abstr^\may_\mathcal{F}(MG))$, where $\Argsd(\mathcal{F})=W\subseteq\Vars$, $\Argsr(\mathcal{F})=Z$, $Z\cap \Vars=\varnothing$. %
  Then, for any $V\subseteq (\Vars\setminus W)$ and $AP\subseteq \Cond_{V}\cup\Loc$, the relation $\mathcal{R}\subseteq \States_1\times\States_2$, defined by $(l_1,\eta_1,\upsilon_1)\mathcal{R} (l_2,\eta_2,\upsilon_2)$ iff $l_1=l_2$, $\upsilon_1=\upsilon_2$ and either $(l_1\in\Scope) \wedge (\eta_1(V)=\eta_2(V))$, or $\eta_1(\Vars)=\eta_2(\Vars)$,
  is the timed simulation over $AP$ between $M_1$ and $M_2$.%
\end{theorem}

\begin{proof}
    First, we show that condition (i) from \cref{df:model-sim} holds.
    By construction of $\abstr^\may_\mathcal{F}(MG)$ from~\cref{alg:general-abstraction}, an abstract initial condition is of the form $g_0\wedge g_{Z_0}$ (cf.~\cref{alg:gen-abstr-l1}), where $g_{Z_0}$ determines for $Z$ its initial value $Z_0:=f(\restr{\sset{g_0}}{W})(Z)$, that is $g_{Z_0}\equiv (Z{=}Z_0)$ (cf.~\cref{alg:gen-abstr-l0}).
    Let $\sini_i=(l_0, \eta_i, \upsilon_i)$, where $\forall_{0\leq j \leq n_{\Clocks}}\upsilon_i(x_j)=0$ for $i=1,2$, and $\eta_1\in \sset{g_0}$, $\eta_2\in\sset{g_0\wedge g_{Z_0}}$.
    Observe that $\eta_2\in\sset{g_0}$ also holds, that is $\eta_1(\Vars)=\eta_2(\Vars)$ (and therefore also $\eta_1(V)=\eta_2(V)$), and thus $(\sini_1,\sini_2)\in\mathcal{R}$.

    Next, we show that condition (ii) from \cref{df:model-sim} holds as well. 
    Note that (ii-a) trivially holds due to the considered propositions in $AP$ ranging over $V\subseteq(\Vars\setminus W)$ only, and the fact that $\mathcal{R}$ implies that related states agree on the evaluations for $V$.
    Therefore, for any $(s_1,s_2)\in\mathcal{R}$ we have that $L_1(s_1)\cap AP = L_2(s_1)\cap AP$, and so only (ii-b) remains to be shown. 

    Let $s_i,s'_i\in\States_i$, $s_i=(l_i,\eta_i,\upsilon_i)$, $s'_i=(l'_i,\eta'_i,\upsilon'_i)$ for $i=1,2$, $(s_1,s_2)\in\mathcal{R}$ and $s_1\longrightarrow_1 s'_1$.
    From \cref{def:t-model}, this must be either a delay-transition or an action transition.

    In the former case, $s_1\xlongrightarrow{{\delta}}_1s'_1$ for some $\delta\in\Rplus$, s.t. (by~\cref{def:t-model}) $l_1=l'_1$, $\eta_1=\eta'_1$, $\upsilon'_1=\upsilon_1+\delta$, and  $\upsilon_1,\upsilon_1+\delta\in \sset{\Inv_1(l_1)}$. 
    From $(s_1,s_2)\in\mathcal{R}$ we know that $l_1=l_2$ and $\upsilon_1=\upsilon_2$.
    From~\cref{alg:general-abstraction} it follows that $\Inv_1=\Inv_2$, implying that   
    $\upsilon_2,\upsilon_2+\delta \in\sset{\Inv_2(l_1)}$, so there must exist a (corresponding) delay-transition $s_2\xlongrightarrow{{\delta}}_2s'_2$, where $\upsilon'_2=\upsilon_2 +\delta$, $l_2=l'_2$, $\eta_2=\eta'_2$.  Hence, we have $(s_2,s'_2)\in\mathcal{R}$.

    In the latter case, $s_1\xlongrightarrow{\alpha_1}_1 s'_1$ for some $e_1=l_1\xhookrightarrow{g_1,\clc_1,-,\alpha_1,X_1}_1l'_1$, s.t. $\eta_1\in\sset{g_1}$, $\upsilon_1\in\sset{\clc_1}$,  $\eta'_1=\Effect(\alpha_1,\eta_1)$, $\upsilon'_1\in\sset{\Inv(l'_1)}$ and $\upsilon'_1=\upsilon_1[X_1=0]$.
    By~\cref{alg:general-abstraction}, each (concrete) labelled edge $e=(l,g,\clc,\varsigma,\alpha,X,l')\in E$ (cf.~\cref{alg:gen-abstr-main-for1}) is associated with the set of matching (abstract) labelled edges 
    $\set{e_a(\eta(W)) \mid \eta \in d(l)}\subseteq\Edges_a$ (cf.~\cref{alg:gen-abstr-inner-for1}) of the form $e_a(k)=(l,g[W=k],\clc,\varsigma,\alpha^\mathit{post}(k)\updcomp\alpha\updcomp\alpha^\mathit{pre}(k),X,l')$ 
    (cf.~\cref{alg:gen-abstr-main-for2}), %
    such that:
    \begin{itemize}
      \item $\alpha^\mathit{pre}(k)$ corresponds to the consecutive assignment statements $Z:=Z_0$ and $W:=k$ if $l\in \Scope$ (cf.~\cref{alg:gen-abstr-pre-if1,alg:gen-abstr-pre-if2,alg:gen-abstr-pre-if3}), and to $\tau$ (``do nothing'' instruction) otherwise, 
      \item $\alpha^\mathit{post}(k)$ corresponds to the consecutive assignment statements $Z:=(f(\restr{\eta}{W}))(Z)$ and $W:=\eta_0(W)$ for $\eta_0\in\sset{g_0}$, if $l'\in \Scope$ (cf.~\cref{alg:gen-abstr-post-if1,alg:gen-abstr-post-if2,alg:gen-abstr-post-if3}), and to $\tau$ otherwise.
    \end{itemize} 
    Since $(s_1,s_2)\in\mathcal{R}$, $\upsilon_1\in\sset{\clc}$, and $\upsilon'_1\in\sset{\Inv_1(l'_1)}$, it follows that $\upsilon_2\in\sset{\clc}$ and  $\upsilon_2[X_1=0]\in\sset{\Inv_2(l'_1)}$. 
    Furthermore, there must exist an action-transition $s_2\xlongrightarrow{\alpha_2}_2
    s'_2$ induced by a labelled edge $e_2=e_a(\eta_1(W))=(l_2,g_1
    [W=\eta_1(W)],\clc,-,\alpha_1^\mathit{post}(\eta_1(W))\updcomp\alpha_1\updcomp\alpha_1^\mathit{pre}(\eta_1(W)),X_1,l'_2)$ matching to $e_1$, where $l_1=l_2$, $l'_1=l'_2$, 
    (by~\cref{over-app-corr} it must be that $\eta_1\in d(l_1)$).
    Since $\eta_1\in\sset{g}$ and $\eta_1(V)=\eta_2(V)$, it follows that $\eta_2\in\sset{g[W=\eta_1(W)]}$. Moreover, given that $\Effect(\alpha^\mathit{post}(\eta_1(W))\updcomp\alpha\updcomp\alpha^\mathit{pre}(\eta_1(W)),\eta_2)=\eta'_2$, we have $\restr{\eta'_2}{\Vars}[W=\eta'_1(W)]={\Effect(\alpha,\restr{\eta_2}{\Vars}[W=\eta_1(W)])}$, or, in other words, $\eta'_1(V)=\eta'_2(V)$.
    Hence, we conclude that $(s_2,s'_2)\in\mathcal{R}$.
\end{proof}

\newcommand{\tablescale}{1.6\columnwidth}

{\begin{table}[!t]
  \setlength{\abovecaptionskip}{-1.7ex}
  \setlength{\belowcaptionskip}{3pt}
  \centering
  \resizebox{\columnwidth}{!}{%
  \setlength{\tabcolsep}{3pt}
\begin{tabular}{@{}|c|c||rrr|rrr|rrr|rrr|@{}}
    \hline
    \multirow{2}{*}{{\small\textbf{V}}} & \multirow{2}{*}{{\small\textbf{C}}} & \multicolumn{3}{c|}{Concrete} & \multicolumn{3}{c|}{A1} & \multicolumn{3}{c|}{A2} & \multicolumn{3}{c|}{A3}\\ \hhline{|~|~||---|---|---|---|}
    && \multicolumn{1}{c|}{\textbf{\#St}} & \multicolumn{1}{c|}{\textbf{M}} & \multicolumn{1}{c|}{\textbf{t}} & \multicolumn{1}{c|}{\textbf{\#St}} & \multicolumn{1}{c|}{\textbf{M}} & \multicolumn{1}{c|}{\textbf{t}}& \multicolumn{1}{c|}{\textbf{\#St}} & \multicolumn{1}{c|}{\textbf{M}} & \multicolumn{1}{c|}{\textbf{t}} & \multicolumn{1}{c|}{\textbf{\#St}} & \multicolumn{1}{c|}{\textbf{M}} & \multicolumn{1}{c|}{\textbf{t}} \\ \hline
    1 & 1 & 1.1e+2 & 9 & 0 & 1.1e+2 & 9 & 0 & 3.9e+1 & 9 & 0 & 3.9e+1 & 9 & 0 \\
    1 & 2 & 1.4e+2 & 9 & 0 & 1.4e+2 & 9 & 0 & 4.9e+1 & 9 & 0 & 4.9e+1 & 9 & 0 \\
    1 & 3 & 1.7e+2 & 9 & 0 & 1.7e+2 & 9 & 0 & 5.9e+1 & 9 & 0 & 5.9e+1 & 9 & 0 \\
    2 & 1 & 5.5e+3 & 9 & 0 & 5.5e+3 & 9 & 0 & 7.4e+2 & 9 & 0 & 4.7e+2 & 9 & 0 \\
    2 & 2 & 9.0e+3 & 9 & 0 & 9.0e+3 & 9 & 0 & 1.2e+3 & 9 & 0 & 7.8e+2 & 9 & 0 \\
    2 & 3 & 1.3e+4 & 10 & 0 & 1.3e+4 & 10 & 0 & 1.7e+3 & 9 & 0 & 1.2e+3 & 10 & 0 \\
    3 & 1 & 2.7e+5 & 33 & 1 & 2.7e+5 & 34 & 1 & 1.4e+4 & 10 & 0 & 5.3e+3 & 10 & 0 \\
    3 & 2 & 5.8e+5 & 61 & 2 & 5.8e+5 & 61 & 2 & 2.7e+4 & 11 & 0 & 1.2e+4 & 11 & 0 \\
    3 & 3 & 1.0e+6 & 100 & 3 & 1.0e+6 & 100 & 3 & 4.7e+4 & 13 & 0 & 2.2e+4 & 12 & 0 \\
    4 & 1 & 1.3e+7 & 1\,182 & 59 & 1.3e+7 & 1\,182 & 59 & 2.5e+5 & 30 & 1 & 5.8e+4 & 15 & 0 \\
    4 & 2 & 3.6e+7 & 3\,247 & 163 & 3.6e+7 & 3\,248 & 162 & 6.1e+5 & 63 & 3 & 1.7e+5 & 26 & 2 \\
    4 & 3 & 7.9e+7 & 7\,113 & 369 & 7.9e+7 & 7\,113 & 369 & 1.3e+6 & 122 & 6 & 4.0e+5 & 45 & 5 \\
    5 & 1 & \multicolumn{3}{c|}{\memout} & \multicolumn{3}{c|}{\memout} & 4.4e+6 & 396 & 24 & 6.2e+5 & 64 & 5 \\
    5 & 2 & \multicolumn{3}{c|}{\memout} & \multicolumn{3}{c|}{\memout} & 1.4e+7 & 1\,187 & 78 & 2.4e+6 & 225 & 31 \\
    5 & 3 & \multicolumn{3}{c|}{\memout} & \multicolumn{3}{c|}{\memout} & 3.5e+7 & 3\,079 & 208 & 7.1e+6 & 621 & 118 \\
    6 & 1 & \multicolumn{3}{c|}{\memout} & \multicolumn{3}{c|}{\memout} & 7.7e+7 & 6\,738 & 541 & 6.3e+6 & 554 & 67 \\
    6 & 2 & \multicolumn{3}{c|}{\memout} & \multicolumn{3}{c|}{\memout} & 3.0e+8 & 26\,554 & 2\,306 & 3.4e+7 & 3\,002 & 523 \\
    6 & 3 & \multicolumn{3}{c|}{\memout} & \multicolumn{3}{c|}{\memout} & \multicolumn{3}{c|}{\memout} & 1.2e+8 & 10\,528 & 2\,542 \\
    7 & 1 & \multicolumn{3}{c|}{\memout} & \multicolumn{3}{c|}{\memout} & \multicolumn{3}{c|}{\memout} & 6.4e+7 & 5\,408 & 821 \\
    7 & 2 & \multicolumn{3}{c|}{\memout} & \multicolumn{3}{c|}{\memout} & \multicolumn{3}{c|}{\memout} & \multicolumn{3}{c|}{\memout} \\
    \hline
\end{tabular}
}
  \caption{Experimental results for model checking $\varphi_1$ (FAA) on concrete and abstract models with no re-voting}
  \label{tab:results-tasv1}
\end{table}}

\section{Experimental Evaluation}
\label{sec:experiments}

In this section we report the series of experiments on a voting system case study.

\subsection{Benchmark: Estonian Internet Voting}

We extend the Continuous-time Asynchronous Multi-Agent System (CAMAS) model of the voting scenario from~\cite{Arias23stratTCTL}, which was inspired by the election procedure in Estonia~\cite{springall2014security}.
In particular, we add a malicious agent (Coercer), extra variables for the Election Authority (voting frequency and tallying), extra locations and labelled transitions for the Voters (interaction with coercer and possible re-voting).
Furthermore, we parameterize the system with the number of voters (\texttt{NV}), the number of candidates (\texttt{NC}), the Boolean specifying whether re-voting is allowed (\texttt{RV}), the type of coercer's behaviour (\texttt{CTYPE}) and punishment criteria (\texttt{OBEY,DISOBEY}).

We use Uppaal model checker~\cite{uppaal2002software} and verify various configurations of the system --- determined by its parameter values --- with regard to the exposure of voters to the coercion through forced abstention or forced participation~\cite{jamroga2024you}.

The voting scenario is standard: each voter (V) first needs to register for the preferred
voting modality: postal vote, e-vote over the internet or traditional paper vote at
a polling station; if time constraints for the chosen modality are met, the election authority (EA) accepts the registration and immediately provides appropriate voting material to that voter (e.g., election package, e-voting credentials, address of the assigned election commission office).
Upon receiving these materials, V proceeds with either casting the vote for selected
candidate, casting an invalid vote (e.g., by crossing more than one candidate) or
abstaining from voting; if time constraints for casting the vote over the given medium are met, then EA records the vote in the tally.
Then, V can interact with the coercer (C) once, and either get punished
for not complying with the instructions or not (and possibly get rewarded).
Finally, when the election time is over, EA closes the vote and C punishes all the V who did not show how they voted beforehand.

As in~\cite{Arias23stratTCTL}, we assign each voting modality a specific time frame: 1--7 for postal vote, 6--9 for e-vote, 10--11 for paper vote, and close the election at 11 time units.

The global (shared) variables:
\begin{itemize}
  \item $\texttt{sh,prev}$: auxiliary variables used to pass the current and previously cast value of $\texttt{vote}$ from V to EA,
\end{itemize}

The Election Authority timed agent graph's local variables:
\begin{itemize}
  \item $\texttt{tally}$: an array of size $\texttt{NC+1}$ storing the number of votes cast per candidate,\footnote{A greater size was used for technical reasons and to improve the readability; note that in practice $\texttt{tally[0]=0}$ is the global invariant.}
  \item $\texttt{freq}$: voting frequency as the number of voters who
\end{itemize}

The Voter timed agent graph's local variables:
\begin{itemize}
  \item $\texttt{mode}$: (currently) chosen voting modality,
  \item $\texttt{vote}$: encoding of cast vote corresponding to the candidate ($\texttt{1..NC}$), invalid vote ($\texttt{0}$) or no previously registered participation ($\texttt{-1}$).
  \item $\texttt{p}$: Boolean variable whether V was punished,
  \item $\texttt{np}$: Boolean variable whether V was not punished,%
\footnote{Indeed, the pair of variables $\texttt{p}$ and $\texttt{np}$ is almost dual, however having both allows to distinguish the (initial) case when C has not yet decided whether to punish V or not.} %
\end{itemize}
In our model we distinguish two types of coercers with deterministic punishment/reward condition based on the shown receipt (here, any proof of how/whether V voted).
The \texttt{TYPE1} coercer will punish a voter only when the \texttt{DISOBEY} condition
matches the shown receipt (or voter refused to show it), whereas \texttt{TYPE2} will
always punish voter except when the \texttt{OBEY} condition matches the receipt.

In particular, the forced abstention attack (FAA) is captured by \texttt{TYPE2} coercer with \texttt{OBEY=-1} (where \texttt{-1} represents that voter did not cast her vote and thus was not counted towards voting frequency)%
; similarly, the forced participation attack (FPA) is captured by \texttt{TYPE1} coercer with \texttt{DISOBEY=-1}.

The FAA and FPA properties can be represented as follows:
\begin{equation}\leqnomode
  \tag{$\varphi_1$}
  \Apath\Always\left(\, V.np=\top \quad\Rightarrow\quad V.voted=\texttt{OBEY}\,\right)
\end{equation}%
\begin{equation}\leqnomode
  \tag{$\varphi_2$}
  \Apath\Always\left(\, V.np=\top \quad\Rightarrow\quad V.voted\neq \texttt{DISOBEY}\,\right)
\end{equation}%
$\varphi_1$ says that in all executions, V can avoid punishment only by abstaining from the voting (as instructed by C). $\varphi_2$ says that in all executions, V must take part in the voting to avoid getting punished.

{\begin{table}[!t]
  \setlength{\abovecaptionskip}{-1.7ex} 
  \setlength{\belowcaptionskip}{3pt}
  \centering
  \resizebox{\columnwidth}{!}{%
  \setlength{\tabcolsep}{3pt}
\begin{tabular}{@{}|c|c|rrr|rrr|rrr|rrr|@{}}
    \hline
    \multirow{2}{*}{{\small\textbf{V}}} & \multirow{2}{*}{{\small\textbf{C}}} & \multicolumn{3}{c|}{Concrete} & \multicolumn{3}{c|}{A1} & \multicolumn{3}{c|}{A2} & \multicolumn{3}{c|}{A3}\\ \hhline{|~|~|---|---|---|---|}
    && \multicolumn{1}{c|}{\textbf{\#St}} & \multicolumn{1}{c|}{\textbf{M}} & \multicolumn{1}{c|}{\textbf{t}} & \multicolumn{1}{c|}{\textbf{\#St}} & \multicolumn{1}{c|}{\textbf{M}} & \multicolumn{1}{c|}{\textbf{t}}& \multicolumn{1}{c|}{\textbf{\#St}} & \multicolumn{1}{c|}{\textbf{M}} & \multicolumn{1}{c|}{\textbf{t}} & \multicolumn{1}{c|}{\textbf{\#St}} & \multicolumn{1}{c|}{\textbf{M}} & \multicolumn{1}{c|}{\textbf{t}}\\ \hline
    1 & 1 & 1.1e+2 & 9 & 0 & 1.1e+2 & 9 & 0 & 4.1e+1 & 9 & 0 & 4.1e+1 & 9 & 0 \\
    1 & 2 & 1.5e+2 & 9 & 0 & 1.5e+2 & 9 & 0 & 5.3e+1 & 9 & 0 & 5.3e+1 & 9 & 0 \\
    1 & 3 & 1.9e+2 & 9 & 0 & 1.9e+2 & 9 & 0 & 6.5e+1 & 9 & 0 & 6.5e+1 & 9 & 0 \\
    2 & 1 & 6.1e+3 & 9 & 0 & 6.1e+3 & 9 & 0 & 8.2e+2 & 9 & 0 & 4.9e+2 & 9 & 0 \\
    2 & 2 & 1.1e+4 & 9 & 0 & 1.1e+4 & 9 & 0 & 1.4e+3 & 9 & 0 & 8.4e+2 & 10 & 0 \\
    2 & 3 & 1.7e+4 & 10 & 0 & 1.7e+4 & 10 & 0 & 2.1e+3 & 9 & 0 & 1.3e+3 & 10 & 0 \\
    3 & 1 & 3.3e+5 & 38 & 1 & 3.3e+5 & 38 & 1 & 1.6e+4 & 10 & 0 & 5.6e+3 & 10 & 0 \\
    3 & 2 & 7.5e+5 & 75 & 2 & 7.5e+5 & 75 & 2 & 3.5e+4 & 12 & 0 & 1.3e+4 & 11 & 0 \\
    3 & 3 & 1.4e+6 & 137 & 4 & 1.4e+6 & 137 & 4 & 6.4e+4 & 14 & 0 & 2.4e+4 & 12 & 0 \\
    4 & 1 & 1.7e+7 & 1\,533 & 73 & 1.7e+7 & 1\,534 & 72  & 3.1e+5 & 36 & 1 & 6.1e+4 & 15 & 0 \\
    4 & 2 & 5.2e+7 & 4\,535 & 218 & 5.2e+7 & 4\,535 & 217 & 8.6e+5 & 83 & 3 & 1.9e+5 & 27 & 2 \\
    4 & 3 & 1.2e+8 & 10\,743 & 531 & 1.2e+8 & 10\,743 & 528 & 1.9e+6 & 176 & 8 & 4.5e+5 & 49 & 6 \\
    5 & 1 & \multicolumn{3}{c|}{\memout} & \multicolumn{3}{c|}{\memout} & 5.8e+6 & 507 & 30 & 6.5e+5 & 67 & 6 \\
    5 & 2 & \multicolumn{3}{c|}{\memout} & \multicolumn{3}{c|}{\memout} & 2.1e+7 & 1\,841 & 113 & 2.7e+6 & 243 & 34 \\
    5 & 3 & \multicolumn{3}{c|}{\memout} & \multicolumn{3}{c|}{\memout} & 5.9e+7 & 5\,020 & 325 & 7.9e+6 & 687 & 131 \\
    6 & 1 & \multicolumn{3}{c|}{\memout} & \multicolumn{3}{c|}{\memout} & 1.1e+8 & 9\,170 & 722 & 6.7e+6 & 583 & 71 \\
    6 & 2 & \multicolumn{3}{c|}{\memout} & \multicolumn{3}{c|}{\memout} & \multicolumn{3}{c|}{\memout} & 3.7e+7 & 3\,253 & 570 \\
    6 & 3 & \multicolumn{3}{c|}{\memout} & \multicolumn{3}{c|}{\memout} & \multicolumn{3}{c|}{\memout} & 1.4e+8 & 12\,190 & 2\,843 \\
    7 & 1 & \multicolumn{3}{c|}{\memout} & \multicolumn{3}{c|}{\memout} & \multicolumn{3}{c|}{\memout} & 6.8e+7 & 5\,964 & 874 \\
    7 & 2 & \multicolumn{3}{c|}{\memout} & \multicolumn{3}{c|}{\memout} & \multicolumn{3}{c|}{\memout} & \multicolumn{3}{c|}{\memout} \\
    \hline
\end{tabular}
 }
  \caption{Experimental results for model checking $\varphi_2$ (FPA) on concrete and abstract models with no re-voting}
  \label{tab:results-tasv2}
\end{table}}

\subsection{Experiments and Results}

For the experiments, we used a modified version of the open-source tool \textsc{EasyAbstract}\footnote{\url{https://tinyurl.com/EasyAbstract4Uppaal}}, which implements the algorithms from~\cite{Jamroga23easyabstract} for Uppaal, to automate generation of the abstract models for each configuration of the system.
Furthermore, we used a coarser variant of the local domain over-approximation based on the agent graph (or its template), where all synchronisation labels are simply discarded.
By doing so, we were able to further reduce the memory and time usage.
Due to FAA and FPA properties relating to the similar subset of atomic propositions, we employed the same abstractions in both cases:
\begin{itemize}
  \item[A1:] removes variables \texttt{tally,freq} in Election Authority TAG;
  \item[A2:] in addition to A1 removes variable \texttt{mode} in Voter TAG(s);
  \item[A3:] in addition to A2 removes variables \texttt{voted,p,np} in all Voter TAGs except one.
\end{itemize}

We report experimental results in the~\cref{tab:results-tasv1,tab:results-tasv2,tab:results-tasv3}.\footnote{%
  The verification was performed on the machine with AMD EPYC 7302P 16-Core 1.5 GHz CPU, 32 GB RAM, Ubuntu 22.04, running \texttt{verifyta} command-line utility from Uppaal v4.1.24 distribution.
  The source code of the models and auxiliary scripts for running verification can be found on: \url{https://github.com/aamas2025submission}.
} 
The first two columns indicate the configuration, that is the number of voters (``V'') and candidates (``C'');
next, the two groups of columns aggregate the details of verification of the forced abstention attack ( ``FAA'') and forced participation attack (``FPA'')  against the concrete and abstract (``A2'' and ``A3'') TMAS graphs, within each group, when property is satisfied, the column ``\#St'' indicates the number of symbolic states (as defined by Uppaal), ``M'' the amount of RAM used (in MiB\footnote{1 MiB = 2$^{20}$ Bytes, for more details see~\cite{IEC60027}.}), ``t'' the time spent by CPU (in sec\extended{onds}, rounded to the nearest whole number).%
\footnote{The time spent on computing the abstract TMAS graph was negligible (< 1s) in all cases considered and is thus not included in the table. }
When model checking $\varphi_1$ (FAA) on models with re-voting allows, the verifier was always returning a counter-example run within <1 sec time. %

It is noteworthy that with the help of abstraction, we were able to almost double
the number of agents in the configuration before running out of memory due to the
state space explosion  (from 3--4 to 6--7 voters), effectively reducing the use of
memory and time resources by up to two orders of magnitude.
Note also that the effectiveness of the method heavily depended on the choice of the variables to remove -- for instance, removing variables \texttt{tally,freq} (abstraction A1) did not prove useful at all.

{\begin{table}[!t]
  \setlength{\abovecaptionskip}{-1.7ex} 
  \setlength{\belowcaptionskip}{3pt}
  \centering
  \resizebox{\columnwidth}{!}{%
  \setlength{\tabcolsep}{3pt}
\begin{tabular}{@{}|c|c|rrr|rrr|rrr|rrr|@{}}
    \hline
    \multirow{2}{*}{{\small\textbf{V}}} & \multirow{2}{*}{{\small\textbf{C}}} & \multicolumn{3}{c|}{Concrete} & \multicolumn{3}{c|}{A1} & \multicolumn{3}{c|}{A2} & \multicolumn{3}{c|}{A3}\\ \hhline{|~|~|---|---|---|---|}
    && \multicolumn{1}{c|}{\textbf{\#St}} & \multicolumn{1}{c|}{\textbf{M}} & \multicolumn{1}{c|}{\textbf{t}} & \multicolumn{1}{c|}{\textbf{\#St}} & \multicolumn{1}{c|}{\textbf{M}} & \multicolumn{1}{c|}{\textbf{t}}& \multicolumn{1}{c|}{\textbf{\#St}} & \multicolumn{1}{c|}{\textbf{M}} & \multicolumn{1}{c|}{\textbf{t}} & \multicolumn{1}{c|}{\textbf{\#St}} & \multicolumn{1}{c|}{\textbf{M}} & \multicolumn{1}{c|}{\textbf{t}} \\ \hline
    1 & 1 & 2.7e+02 & 9 & 0 & 2.7e+02 & 9 & 0 & 1.0e+02 & 9 & 0 & 1.0e+02 & 9 & 0 \\
    1 & 2 & 3.7e+02 & 9 & 0 & 3.7e+02 & 9 & 0 & 1.4e+02 & 9 & 0 & 1.4e+02 & 9 & 0 \\
    1 & 3 & 4.6e+02 & 9 & 0 & 4.6e+02 & 9 & 0 & 1.8e+02 & 9 & 0 & 1.8e+02 & 9 & 0 \\
    2 & 1 & 3.4e+04 & 11 & 0 & 3.4e+04 & 12 & 0 & 5.1e+03 & 9 & 0 & 1.9e+03 & 9 & 0 \\
    2 & 2 & 6.3e+04 & 14 & 0 & 6.3e+04 & 14 & 0 & 9.6e+03 & 10 & 0 & 3.6e+03 & 10 & 0 \\
    2 & 3 & 1.0e+05 & 17 & 0 & 1.0e+05 & 17 & 0 & 1.6e+04 & 10 & 0 & 5.9e+03 & 10 & 0 \\
    3 & 1 & 4.2e+06 & 355 & 16 & 4.2e+06 & 355 & 16 & 2.4e+05 & 29 & 1 & 3.2e+04 & 12 & 0 \\
    3 & 2 & 1.1e+07 & 922 & 44 & 1.1e+07 & 922 & 43 & 6.3e+05 & 62 & 2 & 8.6e+04 & 17 & 1 \\
    3 & 3 & 2.2e+07 & 1\,877 & 98 & 2.2e+07 & 1\,877 & 95 & 1.3e+06 & 119 & 6 & 1.8e+05 & 26 & 3 \\
    4 & 1 & \multicolumn{3}{c|}{\memout} & \multicolumn{3}{c|}{\memout} & 1.1e+07 & 936 & 54 & 5.2e+05 & 53 & 4 \\
    4 & 2 & \multicolumn{3}{c|}{\memout} & \multicolumn{3}{c|}{\memout} & 3.9e+07 & 3\,356 & 217 & 1.9e+06 & 173 & 27 \\
    4 & 3 & \multicolumn{3}{c|}{\memout} & \multicolumn{3}{c|}{\memout} & 1.0e+08 & 8\,640 & 631 & 5.3e+06 & 461 & 109 \\
    5 & 1 & \multicolumn{3}{c|}{\memout} & \multicolumn{3}{c|}{\memout} & \multicolumn{3}{c|}{\memout} & 8.1e+06 & 679 & 88 \\
    5 & 2 & \multicolumn{3}{c|}{\memout} & \multicolumn{3}{c|}{\memout} & \multicolumn{3}{c|}{\memout} & 4.3e+07 & 3\,647 & 750 \\
    5 & 3 & \multicolumn{3}{c|}{\memout} & \multicolumn{3}{c|}{\memout} & \multicolumn{3}{c|}{\memout} & 1.5e+08 & 12\,821 & 3\,879 \\
    6 & 1 & \multicolumn{3}{c|}{\memout} & \multicolumn{3}{c|}{\memout} & \multicolumn{3}{c|}{\memout} & 1.2e+08 & 10\,151 & 1\,617 \\
    6 & 2 & \multicolumn{3}{c|}{\memout} & \multicolumn{3}{c|}{\memout} & \multicolumn{3}{c|}{\memout} & \multicolumn{3}{c|}{\memout} \\
    \hline
\end{tabular}
 }
  \caption{Experimental results for model checking $\varphi_2$ (FPA) on concrete and abstract models with re-voting allowed}
  \label{tab:results-tasv3}
\end{table}}

\section{Conclusions}
\label{sec:conclusion}

In this paper, we propose a new scheme for \emph{agent-based may abstractions of timed MAS}. 
The work extends the recent abstraction method~\cite{Jamroga23variableabstraction}, which was defined only for the untimed case.
We also ``lift'' the algorithms of~\cite{Jamroga23easyabstract} to operate on MAS graphs with clocks, time invariants, and timing guards.

Similarly to~\cite{Jamroga23variableabstraction,Jamroga23easyabstract}, our abstractions transform the specification of the system at the level of timed agent graphs, without ever generating the global model. 
An experimental evaluation, based on a scalable model of Estonian elections, have shown a very promising pattern of results. 
In all cases, computation of the abstract representation by our implementation took negligible time. 
Moreover, it allowed the Uppaal model checker to verify instances with state spaces that are several orders of magnitude larger. 
The experiments showed also that the effectiveness of the method depends on the right selection of variables to be removed; ideally, they should be provided by a domain expert.

In the future, we plan to refine \cref{alg:variable-approx} and \cref{alg:general-abstraction}, so that the approximation of local domain is computed over the pairs of location and zone, where the zone is an abstraction class of clock valuations that satisfy the same set of clock constraints occurring within the model and the formula~\cite{Penczek06advances}. 
Analogously, for the abstract TMAS generation, the labelled edges would be assigned the clock constraints corresponding to the possible zones of the source-location pair.
This way, we hope to obtain more refined abstract models. It remains to be seen if the approach will turn out computationally feasible, or lead to the generation of an enormous (though finite) number of regions.

\begin{acks}
The work has been supported by NCBR Poland and FNR Luxembourg under the PolLux/FNR-CORE
project SpaceVote (POLLUX-XI/14/SpaceVote/2023 and C22/IS/17232062/SpaceVote), the PHC Polonium project MoCcA (BPN/BFR/2023/1/00045),
CNRS IRP ``Le Tr\'{o}jk{\k a}t'', and ANR-22-CE48-0012 project BISOUS.
For the purpose of open access, and in fulfilment of the obligations arising from the grant agreement, the authors have applied a Creative Commons Attribution 4.0 International (CC BY 4.0) license to any Author Accepted Manuscript version arising from this submission.
\end{acks}

\clearpage
\bibliographystyle{ACM-Reference-Format}
\bibliography{other,wojtek-own,wojtek}


\begin{thebibliography}{57}


\ifx \showCODEN    \undefined \def \showCODEN     #1{\unskip}     \fi
\ifx \showDOI      \undefined \def \showDOI       #1{#1}\fi
\ifx \showISBNx    \undefined \def \showISBNx     #1{\unskip}     \fi
\ifx \showISBNxiii \undefined \def \showISBNxiii  #1{\unskip}     \fi
\ifx \showISSN     \undefined \def \showISSN      #1{\unskip}     \fi
\ifx \showLCCN     \undefined \def \showLCCN      #1{\unskip}     \fi
\ifx \shownote     \undefined \def \shownote      #1{#1}          \fi
\ifx \showarticletitle \undefined \def \showarticletitle #1{#1}   \fi
\ifx \showURL      \undefined \def \showURL       {\relax}        \fi
\providecommand\bibfield[2]{#2}
\providecommand\bibinfo[2]{#2}
\providecommand\natexlab[1]{#1}
\providecommand\showeprint[2][]{arXiv:#2}

\bibitem[\protect\citeauthoryear{Alur, Courcoubetis, and Dill}{Alur et~al\mbox{.}}{1993}]%
        {alur1993model}
\bibfield{author}{\bibinfo{person}{Rajeev Alur}, \bibinfo{person}{Costas Courcoubetis}, {and} \bibinfo{person}{David Dill}.} \bibinfo{year}{1993}\natexlab{}.
\newblock \showarticletitle{Model-checking in dense real-time}.
\newblock \bibinfo{journal}{\emph{Information and computation}} \bibinfo{volume}{104}, \bibinfo{number}{1} (\bibinfo{year}{1993}), \bibinfo{pages}{2--34}.
\newblock


\bibitem[\protect\citeauthoryear{Alur, Henzinger, Kupferman, and Vardi}{Alur et~al\mbox{.}}{1998}]%
        {Alur98refinement}
\bibfield{author}{\bibinfo{person}{Rajeev Alur}, \bibinfo{person}{Thomas~A Henzinger}, \bibinfo{person}{Orna Kupferman}, {and} \bibinfo{person}{Moshe~Y Vardi}.} \bibinfo{year}{1998}\natexlab{}.
\newblock \showarticletitle{Alternating refinement relations}. In \bibinfo{booktitle}{\emph{Proceedings of CONCUR}} \emph{(\bibinfo{series}{Lecture Notes in Computer Science}, Vol.~\bibinfo{volume}{1466})}. \bibinfo{pages}{163--178}.
\newblock


\bibitem[\protect\citeauthoryear{Arias, Jamroga, Penczek, Petrucci, and Sidoruk}{Arias et~al\mbox{.}}{2023}]%
        {Arias23stratTCTL}
\bibfield{author}{\bibinfo{person}{Jaime Arias}, \bibinfo{person}{Wojciech Jamroga}, \bibinfo{person}{Wojciech Penczek}, \bibinfo{person}{Laure Petrucci}, {and} \bibinfo{person}{Teofil Sidoruk}.} \bibinfo{year}{2023}\natexlab{}.
\newblock \showarticletitle{Strategic (Timed) {Computation Tree Logic}}. In \bibinfo{booktitle}{\emph{Proceedings of the International Conference on Autonomous Agents and Multiagent Systems, {AAMAS}}}. \bibinfo{publisher}{{ACM}}, \bibinfo{pages}{382--390}.
\newblock
\urldef\tempurl%
\url{https://doi.org/10.5555/3545946.3598661}
\showDOI{\tempurl}


\bibitem[\protect\citeauthoryear{Baier and Katoen}{Baier and Katoen}{2008}]%
        {baier2008principles}
\bibfield{author}{\bibinfo{person}{Christel Baier} {and} \bibinfo{person}{Joost-Pieter Katoen}.} \bibinfo{year}{2008}\natexlab{}.
\newblock \bibinfo{booktitle}{\emph{Principles of model checking}}.
\newblock \bibinfo{publisher}{MIT press}.
\newblock


\bibitem[\protect\citeauthoryear{Ball and Kupferman}{Ball and Kupferman}{2006}]%
        {Ball06abstraction}
\bibfield{author}{\bibinfo{person}{Thomas Ball} {and} \bibinfo{person}{Orna Kupferman}.} \bibinfo{year}{2006}\natexlab{}.
\newblock \showarticletitle{An Abstraction-Refinement Framework for Multi-Agent Systems}. In \bibinfo{booktitle}{\emph{Proceedings of Logic in Computer Science (LICS)}}. \bibinfo{publisher}{{IEEE}}, \bibinfo{pages}{379--388}.
\newblock
\urldef\tempurl%
\url{https://doi.org/10.1109/LICS.2006.10}
\showDOI{\tempurl}


\bibitem[\protect\citeauthoryear{Belardinelli, Condurache, Dima, Jamroga, and Knapik}{Belardinelli et~al\mbox{.}}{2021}]%
        {Belardinelli21bisimulations}
\bibfield{author}{\bibinfo{person}{Francesco Belardinelli}, \bibinfo{person}{Rodica Condurache}, \bibinfo{person}{Catalin Dima}, \bibinfo{person}{Wojciech Jamroga}, {and} \bibinfo{person}{Michal Knapik}.} \bibinfo{year}{2021}\natexlab{}.
\newblock \showarticletitle{Bisimulations for verifying strategic abilities with an application to the {ThreeBallot} voting protocol}.
\newblock \bibinfo{journal}{\emph{Information and Computation}}  \bibinfo{volume}{276} (\bibinfo{year}{2021}), \bibinfo{pages}{104552}.
\newblock
\urldef\tempurl%
\url{https://doi.org/10.1016/j.ic.2020.104552}
\showDOI{\tempurl}


\bibitem[\protect\citeauthoryear{Belardinelli, Kouvaros, and Lomuscio}{Belardinelli et~al\mbox{.}}{2017}]%
        {Belardinelli17dataAware}
\bibfield{author}{\bibinfo{person}{Francesco Belardinelli}, \bibinfo{person}{Panagiotis Kouvaros}, {and} \bibinfo{person}{Alessio Lomuscio}.} \bibinfo{year}{2017}\natexlab{}.
\newblock \showarticletitle{Parameterised Verification of Data-aware Multi-Agent Systems}. In \bibinfo{booktitle}{\emph{Proceedings of {IJCAI}}}. \bibinfo{publisher}{ijcai.org}, \bibinfo{pages}{98--104}.
\newblock
\urldef\tempurl%
\url{https://doi.org/10.24963/ijcai.2017/15}
\showDOI{\tempurl}


\bibitem[\protect\citeauthoryear{Belardinelli and Lomuscio}{Belardinelli and Lomuscio}{2017}]%
        {Belardinelli17abstraction}
\bibfield{author}{\bibinfo{person}{Francesco Belardinelli} {and} \bibinfo{person}{Alessio Lomuscio}.} \bibinfo{year}{2017}\natexlab{}.
\newblock \showarticletitle{Agent-based Abstractions for Verifying Alternating-time Temporal Logic with Imperfect Information}. In \bibinfo{booktitle}{\emph{Proceedings of {AAMAS}}}. \bibinfo{publisher}{{ACM}}, \bibinfo{pages}{1259--1267}.
\newblock


\bibitem[\protect\citeauthoryear{Belardinelli, Lomuscio, and Malvone}{Belardinelli et~al\mbox{.}}{2019}]%
        {Belardinelli19abstractionStrat}
\bibfield{author}{\bibinfo{person}{Francesco Belardinelli}, \bibinfo{person}{Alessio Lomuscio}, {and} \bibinfo{person}{Vadim Malvone}.} \bibinfo{year}{2019}\natexlab{}.
\newblock \showarticletitle{An Abstraction-Based Method for Verifying Strategic Properties in Multi-Agent Systems with Imperfect Information}. In \bibinfo{booktitle}{\emph{Proceedings of {AAAI}}}. \bibinfo{pages}{6030--6037}.
\newblock


\bibitem[\protect\citeauthoryear{Belardinelli, Lomuscio, and Patrizi}{Belardinelli et~al\mbox{.}}{2011}]%
        {Belardinelli11data-abstraction}
\bibfield{author}{\bibinfo{person}{Francesco Belardinelli}, \bibinfo{person}{Alessio Lomuscio}, {and} \bibinfo{person}{Fabio Patrizi}.} \bibinfo{year}{2011}\natexlab{}.
\newblock \showarticletitle{Verification of Deployed Artifact Systems via Data Abstraction}. In \bibinfo{booktitle}{\emph{Proceedings of {ICSOC}}} \emph{(\bibinfo{series}{Lecture Notes in Computer Science}, Vol.~\bibinfo{volume}{7084})}. \bibinfo{publisher}{Springer}, \bibinfo{pages}{142--156}.
\newblock
\urldef\tempurl%
\url{https://doi.org/10.1007/978-3-642-25535-9\_10}
\showDOI{\tempurl}


\bibitem[\protect\citeauthoryear{Bouyer, Fahrenberg, Larsen, Markey, Ouaknine, and Worrell}{Bouyer et~al\mbox{.}}{2018}]%
        {bouyer2018model}
\bibfield{author}{\bibinfo{person}{Patricia Bouyer}, \bibinfo{person}{Uli Fahrenberg}, \bibinfo{person}{Kim~Guldstrand Larsen}, \bibinfo{person}{Nicolas Markey}, \bibinfo{person}{Jo{\"e}l Ouaknine}, {and} \bibinfo{person}{James Worrell}.} \bibinfo{year}{2018}\natexlab{}.
\newblock \showarticletitle{Model checking real-time systems}.
\newblock \bibinfo{journal}{\emph{Handbook of model checking}} (\bibinfo{year}{2018}), \bibinfo{pages}{1001--1046}.
\newblock


\bibitem[\protect\citeauthoryear{Brihaye, Laroussinie, Markey, and Oreiby}{Brihaye et~al\mbox{.}}{2007}]%
        {Brihaye07timedCGS}
\bibfield{author}{\bibinfo{person}{Thomas Brihaye}, \bibinfo{person}{Fran{\c{c}}ois Laroussinie}, \bibinfo{person}{Nicolas Markey}, {and} \bibinfo{person}{Ghassan Oreiby}.} \bibinfo{year}{2007}\natexlab{}.
\newblock \showarticletitle{Timed Concurrent Game Structures}. In \bibinfo{booktitle}{\emph{Proceedings of {CONCUR}}}. \bibinfo{pages}{445--459}.
\newblock
\urldef\tempurl%
\url{https://doi.org/10.1007/978-3-540-74407-8_30}
\showDOI{\tempurl}


\bibitem[\protect\citeauthoryear{Bulling, Dix, and Jamroga}{Bulling et~al\mbox{.}}{2010}]%
        {Bulling10verification}
\bibfield{author}{\bibinfo{person}{Nils Bulling}, \bibinfo{person}{Jurgen Dix}, {and} \bibinfo{person}{Wojciech Jamroga}.} \bibinfo{year}{2010}\natexlab{}.
\newblock \showarticletitle{Model Checking Logics of Strategic Ability: Complexity}.
\newblock In \bibinfo{booktitle}{\emph{Specification and Verification of Multi-Agent Systems}}, \bibfield{editor}{\bibinfo{person}{M.~Dastani}, \bibinfo{person}{K.~Hindriks}, {and} \bibinfo{person}{J.-J. Meyer}} (Eds.). \bibinfo{publisher}{Springer}, \bibinfo{pages}{125--159}.
\newblock


\bibitem[\protect\citeauthoryear{Cimatti, Clarke, Giunchiglia, Giunchiglia, Pistore, Roveri, Sebastiani, and Tacchella}{Cimatti et~al\mbox{.}}{2002}]%
        {Cimatti02nusmv}
\bibfield{author}{\bibinfo{person}{Alessandro Cimatti}, \bibinfo{person}{Edmund~M. Clarke}, \bibinfo{person}{Enrico Giunchiglia}, \bibinfo{person}{Fausto Giunchiglia}, \bibinfo{person}{Marco Pistore}, \bibinfo{person}{Marco Roveri}, \bibinfo{person}{Roberto Sebastiani}, {and} \bibinfo{person}{Armando Tacchella}.} \bibinfo{year}{2002}\natexlab{}.
\newblock \showarticletitle{{NuSMV2}: An Open-Source Tool for Symbolic Model Checking}. In \bibinfo{booktitle}{\emph{Proceedings of Computer Aided Verification (CAV)}} \emph{(\bibinfo{series}{Lecture Notes in Computer Science}, Vol.~\bibinfo{volume}{2404})}. \bibinfo{pages}{359--364}.
\newblock


\bibitem[\protect\citeauthoryear{Clarke and Emerson}{Clarke and Emerson}{1981}]%
        {Clarke81ctl}
\bibfield{author}{\bibinfo{person}{Edmund~M. Clarke} {and} \bibinfo{person}{E.~Allen Emerson}.} \bibinfo{year}{1981}\natexlab{}.
\newblock \showarticletitle{Design and Synthesis of Synchronization Skeletons Using Branching Time Temporal Logic}. In \bibinfo{booktitle}{\emph{Proceedings of Logics of Programs Workshop}} \emph{(\bibinfo{series}{Lecture Notes in Computer Science}, Vol.~\bibinfo{volume}{131})}. \bibinfo{pages}{52--71}.
\newblock


\bibitem[\protect\citeauthoryear{Clarke, Grumberg, Jha, Lu, and Veith}{Clarke et~al\mbox{.}}{2000}]%
        {Clarke00cegar}
\bibfield{author}{\bibinfo{person}{Edmund~M. Clarke}, \bibinfo{person}{Orna Grumberg}, \bibinfo{person}{Somesh Jha}, \bibinfo{person}{Yuan Lu}, {and} \bibinfo{person}{Helmut Veith}.} \bibinfo{year}{2000}\natexlab{}.
\newblock \showarticletitle{Counterexample-Guided Abstraction Refinement}. In \bibinfo{booktitle}{\emph{Proceedings of {CAV}}} \emph{(\bibinfo{series}{Lecture Notes in Computer Science}, Vol.~\bibinfo{volume}{1855})}. \bibinfo{publisher}{Springer}, \bibinfo{pages}{154--169}.
\newblock
\urldef\tempurl%
\url{https://doi.org/10.1007/10722167\_15}
\showDOI{\tempurl}


\bibitem[\protect\citeauthoryear{Clarke, Grumberg, Jha, Lu, and Veith}{Clarke et~al\mbox{.}}{2003}]%
        {Clarke03cegar}
\bibfield{author}{\bibinfo{person}{Edmund~M. Clarke}, \bibinfo{person}{Orna Grumberg}, \bibinfo{person}{Somesh Jha}, \bibinfo{person}{Yuan Lu}, {and} \bibinfo{person}{Helmut Veith}.} \bibinfo{year}{2003}\natexlab{}.
\newblock \showarticletitle{Counterexample-guided abstraction refinement for symbolic model checking}.
\newblock \bibinfo{journal}{\emph{J. {ACM}}} \bibinfo{volume}{50}, \bibinfo{number}{5} (\bibinfo{year}{2003}), \bibinfo{pages}{752--794}.
\newblock
\urldef\tempurl%
\url{https://doi.org/10.1145/876638.876643}
\showDOI{\tempurl}


\bibitem[\protect\citeauthoryear{Clarke, Grumberg, and Long}{Clarke et~al\mbox{.}}{1994}]%
        {Clarke94abstraction}
\bibfield{author}{\bibinfo{person}{Edmund~M. Clarke}, \bibinfo{person}{Orna Grumberg}, {and} \bibinfo{person}{David~E. Long}.} \bibinfo{year}{1994}\natexlab{}.
\newblock \showarticletitle{Model Checking and Abstraction}.
\newblock \bibinfo{journal}{\emph{ACM Transactions on Programming Languages and Systems}} \bibinfo{volume}{16}, \bibinfo{number}{5} (\bibinfo{year}{1994}), \bibinfo{pages}{1512--1542}.
\newblock


\bibitem[\protect\citeauthoryear{Clarke, Henzinger, Veith, and Bloem}{Clarke et~al\mbox{.}}{2018}]%
        {Clarke18principles}
\bibfield{editor}{\bibinfo{person}{Edmund~M. Clarke}, \bibinfo{person}{Thomas~A. Henzinger}, \bibinfo{person}{Helmut Veith}, {and} \bibinfo{person}{Roderick Bloem}} (Eds.). \bibinfo{year}{2018}\natexlab{}.
\newblock \bibinfo{booktitle}{\emph{Handbook of Model Checking}}.
\newblock \bibinfo{publisher}{Springer}.
\newblock
\showISBNx{978-3-319-10574-1}
\urldef\tempurl%
\url{https://doi.org/10.1007/978-3-319-10575-8}
\showDOI{\tempurl}


\bibitem[\protect\citeauthoryear{Cohen, Dam, Lomuscio, and Russo}{Cohen et~al\mbox{.}}{2009}]%
        {Cohen09abstraction-MAS}
\bibfield{author}{\bibinfo{person}{Mika Cohen}, \bibinfo{person}{Mads Dam}, \bibinfo{person}{Alessio Lomuscio}, {and} \bibinfo{person}{Francesco Russo}.} \bibinfo{year}{2009}\natexlab{}.
\newblock \showarticletitle{Abstraction in model checking multi-agent systems}. In \bibinfo{booktitle}{\emph{Proceedings of {(AAMAS}}}. \bibinfo{publisher}{{IFAAMAS}}, \bibinfo{pages}{945--952}.
\newblock


\bibitem[\protect\citeauthoryear{Cousot and Cousot}{Cousot and Cousot}{1977}]%
        {Cousot77abstraction}
\bibfield{author}{\bibinfo{person}{Patrick Cousot} {and} \bibinfo{person}{Radhia Cousot}.} \bibinfo{year}{1977}\natexlab{}.
\newblock \showarticletitle{Abstract Interpretation: {A} Unified Lattice Model for Static Analysis of Programs by Construction or Approximation of Fixpoints}. In \bibinfo{booktitle}{\emph{Conference Record of the Fourth {ACM} Symposium on Principles of Programming Languages}}. \bibinfo{pages}{238--252}.
\newblock
\urldef\tempurl%
\url{https://doi.org/10.1145/512950.512973}
\showDOI{\tempurl}


\bibitem[\protect\citeauthoryear{Dams, Gerth, and Grumberg}{Dams et~al\mbox{.}}{1997}]%
        {Dams97abstraction}
\bibfield{author}{\bibinfo{person}{Dennis Dams}, \bibinfo{person}{Rob Gerth}, {and} \bibinfo{person}{Orna Grumberg}.} \bibinfo{year}{1997}\natexlab{}.
\newblock \showarticletitle{Abstract Interpretation of Reactive Systems}.
\newblock \bibinfo{journal}{\emph{{ACM} Trans. Program. Lang. Syst.}} \bibinfo{volume}{19}, \bibinfo{number}{2} (\bibinfo{year}{1997}), \bibinfo{pages}{253--291}.
\newblock
\urldef\tempurl%
\url{https://doi.org/10.1145/244795.244800}
\showDOI{\tempurl}


\bibitem[\protect\citeauthoryear{Dams and Grumberg}{Dams and Grumberg}{2018}]%
        {Dams18abstraction+refinement}
\bibfield{author}{\bibinfo{person}{Dennis Dams} {and} \bibinfo{person}{Orna Grumberg}.} \bibinfo{year}{2018}\natexlab{}.
\newblock \showarticletitle{Abstraction and Abstraction Refinement}.
\newblock In \bibinfo{booktitle}{\emph{Handbook of Model Checking}}. \bibinfo{publisher}{Springer}, \bibinfo{pages}{385--419}.
\newblock
\urldef\tempurl%
\url{https://doi.org/10.1007/978-3-319-10575-8\_13}
\showDOI{\tempurl}


\bibitem[\protect\citeauthoryear{de~Alfaro, Godefroid, and Jagadeesan}{de~Alfaro et~al\mbox{.}}{2004}]%
        {Alfaro04three}
\bibfield{author}{\bibinfo{person}{Luca de Alfaro}, \bibinfo{person}{Patrice Godefroid}, {and} \bibinfo{person}{Radha Jagadeesan}.} \bibinfo{year}{2004}\natexlab{}.
\newblock \showarticletitle{Three-Valued Abstractions of Games: Uncertainty, but with Precision}. In \bibinfo{booktitle}{\emph{Proceedings of Logic in Computer Science (LICS)}}. \bibinfo{publisher}{IEEE Computer Society}, \bibinfo{pages}{170--179}.
\newblock


\bibitem[\protect\citeauthoryear{De~Bakker, Bergstra, Klop, and Meyer}{De~Bakker et~al\mbox{.}}{1984}]%
        {Bakker84equivalences}
\bibfield{author}{\bibinfo{person}{JW De~Bakker}, \bibinfo{person}{Jan~A. Bergstra}, \bibinfo{person}{Jan~Willem Klop}, {and} \bibinfo{person}{J-J~Ch Meyer}.} \bibinfo{year}{1984}\natexlab{}.
\newblock \showarticletitle{Linear Time and Branching Time Semantics for Recursion with Merge}.
\newblock \bibinfo{journal}{\emph{Theor. Comput. Sci.}}  \bibinfo{volume}{34} (\bibinfo{year}{1984}), \bibinfo{pages}{135--156}.
\newblock
\urldef\tempurl%
\url{https://doi.org/10.1016/0304-3975(84)90114-2}
\showDOI{\tempurl}


\bibitem[\protect\citeauthoryear{Emerson}{Emerson}{1990}]%
        {Emerson90temporal}
\bibfield{author}{\bibinfo{person}{E.~Allen Emerson}.} \bibinfo{year}{1990}\natexlab{}.
\newblock \showarticletitle{Temporal and Modal Logic}.
\newblock In \bibinfo{booktitle}{\emph{Handbook of Theoretical Computer Science}}, \bibfield{editor}{\bibinfo{person}{J.~van Leeuwen}} (Ed.). Vol.~\bibinfo{volume}{B}. \bibinfo{publisher}{Elsevier}, \bibinfo{pages}{995--1072}.
\newblock


\bibitem[\protect\citeauthoryear{Enea and Dima}{Enea and Dima}{2008}]%
        {Enea08abstractions}
\bibfield{author}{\bibinfo{person}{Constantin Enea} {and} \bibinfo{person}{Catalin Dima}.} \bibinfo{year}{2008}\natexlab{}.
\newblock \showarticletitle{Abstractions of multi-agent systems}.
\newblock \bibinfo{journal}{\emph{International Transactions on Systems Science and Applications}} \bibinfo{volume}{3}, \bibinfo{number}{4} (\bibinfo{year}{2008}), \bibinfo{pages}{329--337}.
\newblock


\bibitem[\protect\citeauthoryear{Gerth, Kuiper, Peled, and Penczek}{Gerth et~al\mbox{.}}{1999}]%
        {Gerth99por}
\bibfield{author}{\bibinfo{person}{Rob Gerth}, \bibinfo{person}{Ruurd Kuiper}, \bibinfo{person}{Doron Peled}, {and} \bibinfo{person}{Wojciech Penczek}.} \bibinfo{year}{1999}\natexlab{}.
\newblock \showarticletitle{A Partial Order Approach to Branching Time Logic Model Checking}. In \bibinfo{booktitle}{\emph{Proceedings of {ISTCS}}}. \bibinfo{publisher}{IEEE}, \bibinfo{pages}{130--139}.
\newblock


\bibitem[\protect\citeauthoryear{Godefroid}{Godefroid}{2014}]%
        {Godefroid14abstraction-software}
\bibfield{author}{\bibinfo{person}{Patrice Godefroid}.} \bibinfo{year}{2014}\natexlab{}.
\newblock \showarticletitle{May/Must Abstraction-Based Software Model Checking for Sound Verification and Falsification}.
\newblock In \bibinfo{booktitle}{\emph{Software Systems Safety}}. Vol.~\bibinfo{volume}{36}. \bibinfo{publisher}{{IOS} Press}, \bibinfo{pages}{1--16}.
\newblock
\urldef\tempurl%
\url{https://doi.org/10.3233/978-1-61499-385-8-1}
\showDOI{\tempurl}


\bibitem[\protect\citeauthoryear{Godefroid, Huth, and Jagadeesan}{Godefroid et~al\mbox{.}}{2001}]%
        {Godefroid01abstractionbased}
\bibfield{author}{\bibinfo{person}{Patrice Godefroid}, \bibinfo{person}{Michael Huth}, {and} \bibinfo{person}{Radha Jagadeesan}.} \bibinfo{year}{2001}\natexlab{}.
\newblock \showarticletitle{Abstraction-based model checking using modal transition systems}. In \bibinfo{booktitle}{\emph{Proceedings of CONCUR}} \emph{(\bibinfo{series}{Lecture Notes in Computer Science}, Vol.~\bibinfo{volume}{2154})}. Springer, \bibinfo{pages}{426--440}.
\newblock


\bibitem[\protect\citeauthoryear{Godefroid and Jagadeesan}{Godefroid and Jagadeesan}{2002}]%
        {Godefroid02abstraction}
\bibfield{author}{\bibinfo{person}{Patrice Godefroid} {and} \bibinfo{person}{Radha Jagadeesan}.} \bibinfo{year}{2002}\natexlab{}.
\newblock \showarticletitle{Automatic Abstraction Using Generalized Model Checking}. In \bibinfo{booktitle}{\emph{Proceedings of Computer Aided Verification (CAV)}} \emph{(\bibinfo{series}{Lecture Notes in Computer Science}, Vol.~\bibinfo{volume}{2404})}. \bibinfo{publisher}{Springer}, \bibinfo{pages}{137--150}.
\newblock
\urldef\tempurl%
\url{https://doi.org/10.1007/3-540-45657-0_11}
\showDOI{\tempurl}


\bibitem[\protect\citeauthoryear{Godefroid, Nori, Rajamani, and Tetali}{Godefroid et~al\mbox{.}}{2010}]%
        {Godefroid10yogi}
\bibfield{author}{\bibinfo{person}{Patrice Godefroid}, \bibinfo{person}{Aditya~V. Nori}, \bibinfo{person}{Sriram~K. Rajamani}, {and} \bibinfo{person}{SaiDeep Tetali}.} \bibinfo{year}{2010}\natexlab{}.
\newblock \showarticletitle{Compositional may-must program analysis: unleashing the power of alternation}. In \bibinfo{booktitle}{\emph{Proceedings of {POPL}}}. \bibinfo{publisher}{{ACM}}, \bibinfo{pages}{43--56}.
\newblock
\urldef\tempurl%
\url{https://doi.org/10.1145/1706299.1706307}
\showDOI{\tempurl}


\bibitem[\protect\citeauthoryear{Gu, Jensen, Poulsen, Seceleanu, Enoiu, and Lundqvist}{Gu et~al\mbox{.}}{2022a}]%
        {gu2022verifiable}
\bibfield{author}{\bibinfo{person}{Rong Gu}, \bibinfo{person}{Peter~G Jensen}, \bibinfo{person}{Danny~B Poulsen}, \bibinfo{person}{Cristina Seceleanu}, \bibinfo{person}{Eduard Enoiu}, {and} \bibinfo{person}{Kristina Lundqvist}.} \bibinfo{year}{2022}\natexlab{a}.
\newblock \showarticletitle{Verifiable strategy synthesis for multiple autonomous agents: a scalable approach}.
\newblock \bibinfo{journal}{\emph{International Journal on Software Tools for Technology Transfer}} \bibinfo{volume}{24}, \bibinfo{number}{3} (\bibinfo{year}{2022}), \bibinfo{pages}{395--414}.
\newblock


\bibitem[\protect\citeauthoryear{Gu, Jensen, Seceleanu, Enoiu, and Lundqvist}{Gu et~al\mbox{.}}{2022b}]%
        {gu2022correctness}
\bibfield{author}{\bibinfo{person}{Rong Gu}, \bibinfo{person}{Peter~G Jensen}, \bibinfo{person}{Cristina Seceleanu}, \bibinfo{person}{Eduard Enoiu}, {and} \bibinfo{person}{Kristina Lundqvist}.} \bibinfo{year}{2022}\natexlab{b}.
\newblock \showarticletitle{Correctness-guaranteed strategy synthesis and compression for multi-agent autonomous systems}.
\newblock \bibinfo{journal}{\emph{Science of Computer Programming}}  \bibinfo{volume}{224} (\bibinfo{year}{2022}), \bibinfo{pages}{102894}.
\newblock


\bibitem[\protect\citeauthoryear{Gurfinkel, Wei, and Chechik}{Gurfinkel et~al\mbox{.}}{2006}]%
        {Gurfinkel06yasm}
\bibfield{author}{\bibinfo{person}{Arie Gurfinkel}, \bibinfo{person}{Ou Wei}, {and} \bibinfo{person}{Marsha Chechik}.} \bibinfo{year}{2006}\natexlab{}.
\newblock \showarticletitle{Yasm: {A} Software Model-Checker for Verification and Refutation}. In \bibinfo{booktitle}{\emph{Proceedings of {CAV}}} \emph{(\bibinfo{series}{Lecture Notes in Computer Science}, Vol.~\bibinfo{volume}{4144})}. \bibinfo{publisher}{Springer}, \bibinfo{pages}{170--174}.
\newblock
\urldef\tempurl%
\url{https://doi.org/10.1007/11817963\_18}
\showDOI{\tempurl}


\bibitem[\protect\citeauthoryear{Huang and Van Der~Meyden}{Huang and Van Der~Meyden}{2014}]%
        {Huang14symbolic-epist}
\bibfield{author}{\bibinfo{person}{Xiaowei Huang} {and} \bibinfo{person}{Ron Van Der~Meyden}.} \bibinfo{year}{2014}\natexlab{}.
\newblock \showarticletitle{Symbolic Model Checking Epistemic Strategy Logic}. In \bibinfo{booktitle}{\emph{Proceedings of AAAI Conference on Artificial Intelligence}}. \bibinfo{pages}{1426--1432}.
\newblock


\bibitem[\protect\citeauthoryear{IEC 60027-2}{IEC 60027-2}{2000}]%
        {IEC60027}
IEC 60027-2 \bibinfo{year}{2000}\natexlab{}.
\newblock \bibinfo{booktitle}{\emph{Letter symbols to be used in electrical technology --- {Part} 2: {Telecommunications} and electronics}}.
\newblock \bibinfo{type}{International Standard}. \bibinfo{institution}{{I}nternational {E}lectrotechnical {C}ommission}, \bibinfo{address}{Geneva, Switzerland}.
\newblock
\showISBNx{2831854903}


\bibitem[\protect\citeauthoryear{Jamroga and Kim}{Jamroga and Kim}{2023a}]%
        {Jamroga23variableabstraction}
\bibfield{author}{\bibinfo{person}{Wojciech Jamroga} {and} \bibinfo{person}{Yan Kim}.} \bibinfo{year}{2023}\natexlab{a}.
\newblock \showarticletitle{Practical Abstraction for Model Checking of Multi-Agent Systems}. In \bibinfo{booktitle}{\emph{Proceedings of the 20th International Conference on Principles of Knowledge Representation and Reasoning, {KR}}}. \bibinfo{pages}{384--394}.
\newblock
\urldef\tempurl%
\url{https://doi.org/10.24963/KR.2023/38}
\showDOI{\tempurl}


\bibitem[\protect\citeauthoryear{Jamroga and Kim}{Jamroga and Kim}{2023b}]%
        {Jamroga23easyabstract}
\bibfield{author}{\bibinfo{person}{Wojciech Jamroga} {and} \bibinfo{person}{Yan Kim}.} \bibinfo{year}{2023}\natexlab{b}.
\newblock \showarticletitle{Practical Model Reductions for Verification of Multi-Agent Systems}. In \bibinfo{booktitle}{\emph{Proceedings of the Thirty-Second International Joint Conference on Artificial Intelligence, {IJCAI}}}. \bibinfo{publisher}{ijcai.org}, \bibinfo{pages}{7135--7139}.
\newblock
\urldef\tempurl%
\url{https://doi.org/10.24963/IJCAI.2023/834}
\showDOI{\tempurl}


\bibitem[\protect\citeauthoryear{Jamroga, Kim, Roenne, and Ryan}{Jamroga et~al\mbox{.}}{2024}]%
        {jamroga2024you}
\bibfield{author}{\bibinfo{person}{Wojciech Jamroga}, \bibinfo{person}{Yan Kim}, \bibinfo{person}{Peter~B Roenne}, {and} \bibinfo{person}{Peter~YA Ryan}.} \bibinfo{year}{2024}\natexlab{}.
\newblock \showarticletitle{“You Shall not Abstain!” A Formal Study of Forced Participation}. In \bibinfo{booktitle}{\emph{Proceeding of the 9th Workshop on Advances in Secure Electronic Voting, Voting}}.
\newblock


\bibitem[\protect\citeauthoryear{Jamroga, Penczek, Sidoruk, Dembi{\'n}ski, and Mazurkiewicz}{Jamroga et~al\mbox{.}}{2020}]%
        {Jamroga20POR-JAIR}
\bibfield{author}{\bibinfo{person}{Wojciech Jamroga}, \bibinfo{person}{Wojciech Penczek}, \bibinfo{person}{Teofil Sidoruk}, \bibinfo{person}{Piotr Dembi{\'n}ski}, {and} \bibinfo{person}{Antoni Mazurkiewicz}.} \bibinfo{year}{2020}\natexlab{}.
\newblock \showarticletitle{Towards Partial Order Reductions for Strategic Ability}.
\newblock \bibinfo{journal}{\emph{Journal of Artificial Intelligence Research}}  \bibinfo{volume}{68} (\bibinfo{year}{2020}), \bibinfo{pages}{817--850}.
\newblock
\urldef\tempurl%
\url{https://doi.org/10.1613/jair.1.11936}
\showDOI{\tempurl}


\bibitem[\protect\citeauthoryear{Kacprzak, Lomuscio, and Penczek}{Kacprzak et~al\mbox{.}}{2004}]%
        {Kacprzak04verifying}
\bibfield{author}{\bibinfo{person}{Magdalena Kacprzak}, \bibinfo{person}{Alessio Lomuscio}, {and} \bibinfo{person}{Wojciech Penczek}.} \bibinfo{year}{2004}\natexlab{}.
\newblock \showarticletitle{Verification of Multiagent Systems via Unbounded Model Checking}. In \bibinfo{booktitle}{\emph{Proceedings of {AAMAS}}}. \bibinfo{publisher}{{IEEE} Computer Society}, \bibinfo{pages}{638--645}.
\newblock
\urldef\tempurl%
\url{https://doi.org/10.1109/AAMAS.2004.10086}
\showDOI{\tempurl}


\bibitem[\protect\citeauthoryear{Kouvaros and Lomuscio}{Kouvaros and Lomuscio}{2017}]%
        {Kouvaros17predicateAbstraction}
\bibfield{author}{\bibinfo{person}{Panagiotis Kouvaros} {and} \bibinfo{person}{Alessio Lomuscio}.} \bibinfo{year}{2017}\natexlab{}.
\newblock \showarticletitle{Parameterised Verification of Infinite State Multi-Agent Systems via Predicate Abstraction}. In \bibinfo{booktitle}{\emph{Proceedings of {AAAI}}}. \bibinfo{pages}{3013--3020}.
\newblock


\bibitem[\protect\citeauthoryear{Lomuscio and Penczek}{Lomuscio and Penczek}{2007}]%
        {Lomuscio07tempoepist}
\bibfield{author}{\bibinfo{person}{Alessio Lomuscio} {and} \bibinfo{person}{Wojciech Penczek}.} \bibinfo{year}{2007}\natexlab{}.
\newblock \showarticletitle{Symbolic Model Checking for Temporal-Epistemic Logics}.
\newblock \bibinfo{journal}{\emph{SIGACT News}} \bibinfo{volume}{38}, \bibinfo{number}{3} (\bibinfo{year}{2007}), \bibinfo{pages}{77--99}.
\newblock
\urldef\tempurl%
\url{https://doi.org/10.1145/1324215.1324231}
\showDOI{\tempurl}


\bibitem[\protect\citeauthoryear{Lomuscio, Qu, and Raimondi}{Lomuscio et~al\mbox{.}}{2017}]%
        {lomuscio17mcmas}
\bibfield{author}{\bibinfo{person}{Alessio Lomuscio}, \bibinfo{person}{Hongyang Qu}, {and} \bibinfo{person}{Franco Raimondi}.} \bibinfo{year}{2017}\natexlab{}.
\newblock \showarticletitle{{MCMAS}: An Open-Source Model Checker for the Verification of Multi-Agent Systems}.
\newblock \bibinfo{journal}{\emph{International Journal on Software Tools for Technology Transfer}} \bibinfo{volume}{19}, \bibinfo{number}{1} (\bibinfo{year}{2017}), \bibinfo{pages}{9--30}.
\newblock
\urldef\tempurl%
\url{https://doi.org/10.1007/s10009-015-0378-x}
\showDOI{\tempurl}


\bibitem[\protect\citeauthoryear{Lomuscio, Qu, and Russo}{Lomuscio et~al\mbox{.}}{2010}]%
        {Lomuscio10dataAbstraction}
\bibfield{author}{\bibinfo{person}{Alessio Lomuscio}, \bibinfo{person}{Hongyang Qu}, {and} \bibinfo{person}{Francesco Russo}.} \bibinfo{year}{2010}\natexlab{}.
\newblock \showarticletitle{Automatic Data-Abstraction in Model Checking Multi-Agent Systems}. In \bibinfo{booktitle}{\emph{Model Checking and Artificial Intelligence}} \emph{(\bibinfo{series}{Lecture Notes in Computer Science}, Vol.~\bibinfo{volume}{6572})}. \bibinfo{publisher}{Springer}, \bibinfo{pages}{52--68}.
\newblock
\urldef\tempurl%
\url{https://doi.org/10.1007/978-3-642-20674-0\_4}
\showDOI{\tempurl}


\bibitem[\protect\citeauthoryear{McMillan}{McMillan}{1993}]%
        {McMillan93symbolic-mcheck}
\bibfield{author}{\bibinfo{person}{Kenneth~L. McMillan}.} \bibinfo{year}{1993}\natexlab{}.
\newblock \bibinfo{booktitle}{\emph{Symbolic Model Checking: An Approach to the State Explosion Problem}}.
\newblock \bibinfo{publisher}{Kluwer Academic Publishers}.
\newblock


\bibitem[\protect\citeauthoryear{McMillan}{McMillan}{2002}]%
        {McMillan02unbounded}
\bibfield{author}{\bibinfo{person}{Kenneth~L. McMillan}.} \bibinfo{year}{2002}\natexlab{}.
\newblock \showarticletitle{Applying {SAT} Methods in Unbounded Symbolic Model Checking}. In \bibinfo{booktitle}{\emph{Proceedings of Computer Aided Verification (CAV)}} \emph{(\bibinfo{series}{Lecture Notes in Computer Science}, Vol.~\bibinfo{volume}{2404})}. \bibinfo{pages}{250--264}.
\newblock


\bibitem[\protect\citeauthoryear{Peled}{Peled}{1993}]%
        {Peled93representatives}
\bibfield{author}{\bibinfo{person}{Doron~A. Peled}.} \bibinfo{year}{1993}\natexlab{}.
\newblock \showarticletitle{All from One, One for All: on Model Checking Using Representatives}. In \bibinfo{booktitle}{\emph{Proceedings of {CAV}}} \emph{(\bibinfo{series}{Lecture Notes in Computer Science}, Vol.~\bibinfo{volume}{697})}, \bibfield{editor}{\bibinfo{person}{Costas Courcoubetis}} (Ed.). \bibinfo{publisher}{Springer}, \bibinfo{pages}{409--423}.
\newblock
\urldef\tempurl%
\url{https://doi.org/10.1007/3-540-56922-7\_34}
\showDOI{\tempurl}


\bibitem[\protect\citeauthoryear{Penczek and Lomuscio}{Penczek and Lomuscio}{2003}]%
        {Penczek03ctlk}
\bibfield{author}{\bibinfo{person}{Wojciech Penczek} {and} \bibinfo{person}{Alessio Lomuscio}.} \bibinfo{year}{2003}\natexlab{}.
\newblock \showarticletitle{Verifying Epistemic Properties of Multi-Agent Systems via Bounded Model Checking}. In \bibinfo{booktitle}{\emph{Proceedings of {AAMAS}}} (Melbourne, Australia). \bibinfo{publisher}{ACM Press}, \bibinfo{address}{New York, NY, USA}, \bibinfo{pages}{209--216}.
\newblock
\showISBNx{1-58113-683-8}


\bibitem[\protect\citeauthoryear{Penczek and P{\'o}{\l}rola}{Penczek and P{\'o}{\l}rola}{2001}]%
        {penczek2001abstractions}
\bibfield{author}{\bibinfo{person}{Wojciech Penczek} {and} \bibinfo{person}{Agata P{\'o}{\l}rola}.} \bibinfo{year}{2001}\natexlab{}.
\newblock \showarticletitle{Abstractions and partial order reductions for checking branching properties of time Petri nets}. In \bibinfo{booktitle}{\emph{Applications and Theory of Petri Nets 2001: 22nd International Conference, ICATPN 2001 Newcastle upon Tyne, UK, June 25--29, 2001 Proceedings 22}}. Springer, \bibinfo{pages}{323--342}.
\newblock


\bibitem[\protect\citeauthoryear{Penczek and P{\'{o}}lrola}{Penczek and P{\'{o}}lrola}{2006}]%
        {Penczek06advances}
\bibfield{author}{\bibinfo{person}{Wojciech Penczek} {and} \bibinfo{person}{Agata P{\'{o}}lrola}.} \bibinfo{year}{2006}\natexlab{}.
\newblock \bibinfo{booktitle}{\emph{Advances in Verification of Time {P}etri Nets and Timed Automata: {A} Temporal Logic Approach}}. \bibinfo{series}{Studies in Computational Intelligence}, Vol.~\bibinfo{volume}{20}.
\newblock \bibinfo{publisher}{Springer}.
\newblock
\showISBNx{978-3-540-32869-8}
\urldef\tempurl%
\url{https://doi.org/10.1007/978-3-540-32870-4}
\showDOI{\tempurl}


\bibitem[\protect\citeauthoryear{Schnoebelen}{Schnoebelen}{2003}]%
        {Schnoebelen03complexity}
\bibfield{author}{\bibinfo{person}{Ph. Schnoebelen}.} \bibinfo{year}{2003}\natexlab{}.
\newblock \showarticletitle{The Complexity of Temporal Model Checking}. In \bibinfo{booktitle}{\emph{Advances in Modal Logics, Proceedings of AiML 2002}}. \bibinfo{publisher}{World Scientific}.
\newblock


\bibitem[\protect\citeauthoryear{Shoham and Grumberg}{Shoham and Grumberg}{2004}]%
        {Shoham04abstraction}
\bibfield{author}{\bibinfo{person}{Sharon Shoham} {and} \bibinfo{person}{Orna Grumberg}.} \bibinfo{year}{2004}\natexlab{}.
\newblock \showarticletitle{Monotonic Abstraction-Refinement for {CTL}}. In \bibinfo{booktitle}{\emph{Proceedings of {TACAS}}} \emph{(\bibinfo{series}{Lecture Notes in Computer Science}, Vol.~\bibinfo{volume}{2988})}. \bibinfo{publisher}{Springer}, \bibinfo{pages}{546--560}.
\newblock
\urldef\tempurl%
\url{https://doi.org/10.1007/978-3-540-24730-2\_40}
\showDOI{\tempurl}


\bibitem[\protect\citeauthoryear{Springall, Finkenauer, Durumeric, Kitcat, Hursti, MacAlpine, and Halderman}{Springall et~al\mbox{.}}{2014}]%
        {springall2014security}
\bibfield{author}{\bibinfo{person}{Drew Springall}, \bibinfo{person}{Travis Finkenauer}, \bibinfo{person}{Zakir Durumeric}, \bibinfo{person}{Jason Kitcat}, \bibinfo{person}{Harri Hursti}, \bibinfo{person}{Margaret MacAlpine}, {and} \bibinfo{person}{J~Alex Halderman}.} \bibinfo{year}{2014}\natexlab{}.
\newblock \showarticletitle{Security analysis of the Estonian internet voting system}. In \bibinfo{booktitle}{\emph{Proceedings of the 2014 ACM SIGSAC Conference on Computer and Communications Security}}. \bibinfo{pages}{703--715}.
\newblock


\bibitem[\protect\citeauthoryear{Tripakis and Yovine}{Tripakis and Yovine}{2001}]%
        {tripakis2001analysis}
\bibfield{author}{\bibinfo{person}{Stavros Tripakis} {and} \bibinfo{person}{Sergio Yovine}.} \bibinfo{year}{2001}\natexlab{}.
\newblock \showarticletitle{Analysis of timed systems using time-abstracting bisimulations}.
\newblock \bibinfo{journal}{\emph{Formal Methods in System Design}}  \bibinfo{volume}{18} (\bibinfo{year}{2001}), \bibinfo{pages}{25--68}.
\newblock


\bibitem[\protect\citeauthoryear{{Uppsala University and Aalborg University}}{{Uppsala University and Aalborg University}}{2002}]%
        {uppaal2002software}
\bibfield{author}{\bibinfo{person}{{Uppsala University and Aalborg University}}.} \bibinfo{year}{2002}\natexlab{}.
\newblock \bibinfo{booktitle}{\emph{UPPAAL}}.
\newblock
\urldef\tempurl%
\url{https://uppaal.org}
\showURL{%
\tempurl}


\end{thebibliography}

\end{document}